\documentclass[10pt,journal,epsfig]{IEEEtran}

\usepackage[dvips]{graphicx}
\usepackage{graphicx}
\usepackage{amssymb}
\usepackage{cite}
\usepackage{subfig}
\usepackage{amsmath}
\usepackage{algorithm}
\usepackage{algorithmic}
\usepackage{multirow}
\usepackage{color}

\newcommand{\bm}{\boldsymbol}
\newtheorem{lemma}{Lemma}
\newtheorem{theorem}{Theorem}

\begin{document}

\title{Spatial Channel Covariance Estimation and Two-Timescale Beamforming
for IRS-Assisted Millimeter Wave Systems}

\author{Hongwei Wang, Jun Fang, Huiping
Duan, and Hongbin Li, ~\IEEEmembership{Fellow,~IEEE}
\thanks{Hongwei Wang, and Jun Fang are with the National Key Laboratory
of Science and Technology on Communications, University of
Electronic Science and Technology of China, Chengdu 611731, China,
Email: JunFang@uestc.edu.cn}
\thanks{Huiping Duan is with the School of Information and Communications Engineering,
University of Electronic Science and Technology of China, Chengdu
611731, China, Email: huipingduan@uestc.edu.cn}
\thanks{Hongbin Li is
with the Department of Electrical and Computer Engineering,
Stevens Institute of Technology, Hoboken, NJ 07030, USA, E-mail:
Hongbin.Li@stevens.edu}
\thanks{\copyright 2022 IEEE. Personal use of this material is permitted. Permission from IEEE must be obtained for all other uses, in
any current or future media, including reprinting/republishing this material for advertising or promotional purposes, creating new collective works, for resale or redistribution to servers or lists, or reuse of any copyrighted component of this work in other works.
}
}

\maketitle

\begin{abstract}
We consider the problem of spatial channel covariance matrix (CCM)
estimation for intelligent reflecting surface (IRS)-assisted
millimeter wave (mmWave) communication systems. Spatial CCM is
essential for two-timescale beamforming in IRS-assisted systems;
however, estimating the spatial CCM is challenging due to the
passive nature of reflecting elements and the large size of the
CCM resulting from massive reflecting elements of the IRS. In this
paper, we propose a CCM estimation method by exploiting the
low-rankness as well as the positive semi-definite (PSD) 3-level
Toeplitz structure of the CCM. Estimation of the CCM is formulated
as a semidefinite programming (SDP) problem and an alternating
direction method of multipliers (ADMM) algorithm is developed. Our
analysis shows that the proposed method is theoretically
guaranteed to attain a reliable CCM estimate with a sample
complexity much smaller than the dimension of the CCM. Thus the
proposed method can help achieve a significant training overhead
reduction. Simulation results are presented to illustrate the
effectiveness of our proposed method and the performance of
two-timescale beamforming scheme based on the estimated CCM.
\end{abstract}


\section{Introduction}
Millimeter Wave (mmWave) communication is considered as a
promising technology for future cellular networks due to its
potential to offer gigabits-per-second communication data
rates~\cite{RanganRappaport14}. Nevertheless, due to the small
wavelength, mmWave signals have limited diffraction and scattering
abilities. As a result, mmWave communications are vulnerable to
blockage events, which can be frequent in indoor and dense urban
environments. Intelligent reflecting surface (IRS) has been
recently introduced as a cost-effective and energy-efficient
solution to address the blockage issue for mmWave
communications~\cite{WangFang20-1}. The IRS, also referred to as
reconfigurable intelligent surface (RIS), is a planar array made
of a newly developed metamaterial. It comprises a large number of
reconfigurable passive elements, each of which can independently
reflect the incident signal with a reconfigurable phase shift. By
properly adjusting the phase shifts of the passive elements, IRS
can help realize a programmable and desirable wireless propagation
environment \cite{WuZhang19,ZhengYou2022}.



Channel state information (CSI) acquisition is a pre-requisite to
achieve the full potential of IRS-assisted mmWave systems. There
have been a plethora of studies on how to acquire the
instantaneous CSI (I-CSI) for IRS-assisted mmWave systems.
Specifically, to reduce the training overhead, some works
exploited the inherent sparsity of mmWave channels and developed
compressed sensing-based methods to estimate the cascade
channel~\cite{WangFang20-3,LiuGao20,WeiShen21}. The sparse
scattering characteristics were also utilized to devise fast beam
training/alignment
schemes~\cite{WangZhang21,WangFang22-1,WangFang22-2}, where the
objective is to obtain partial I-CSI to simultaneously align the
beam for both the transmitter-IRS link and the IRS-user link.
Other works, e.g.,~\cite{LinJin21,DeDe21,ZhengWang22}, developed
tensor decomposition-based channel estimation methods by utilizing
some intrinsic multi-dimensional structure of cascade channels.
Despite these efforts, system optimization based on I-CSI is still
considered as a formidable task due to the following difficulties.
First, the coherence time of mmWave channels is drastically
shorter than that of sub-6GHz channels. This implies that channel
estimation and system optimization (i.e. joint active/passive
beamforming) should be performed more frequently, which entails
tremendous computational resources. Second, system optimization
based on I-CSI requires frequent transmissions of control signals
from the base station (BS) to the IRS, which involves a
significant amount of training overhead.


To address the above difficulties, some attempts have been made by
exploiting channel statistics for joint active and passive
beamforming, e.g.,~\cite{YangWang20, HuGao20,
HuWang20,ZhaoWu21,GanZhong21}. Specifically, in~\cite{ZhaoWu21}, a
two-timescale beamforming protocol was proposed for IRS-assisted
systems, where the reflecting coefficients at the IRS are designed
according to the long-term (i.e. statistical) CSI, and the
transmit beamforming matrix is devised based on the instantaneous
equivalent channel in a short-term scale. Statistical CSI, usually
characterized by the spatial CCM, is essential for two-timescale
beamforming in IRS-assisted systems. Obtaining the spatial CCM,
however, is challenging due to the passive nature of reflecting
elements and the large size of the CCM resulting from massive
reflecting elements of the IRS. To our best knowledge, how to
estimate the spatial CCM for IRS-assisted mmWave systems has not
been reported before. Although there are some works on CCM
estimation for conventional mmWave systems, e.g.,
\cite{ParkHeath18,ParkAli19,WangZhang19}, these methods can not be
straightforwardly extended to the IRS-aided systems.


In this paper, we propose a CCM estimation method for IRS-assisted
mmWave systems. The proposed method exploits the low-rankness as
well as the positive semi-definite (PSD) 3-level Toeplitz
structure of the CCM. We formulate the estimation problem as a
semidefinite programming (SDP) problem which is solved by an
alternating direction method of multipliers (ADMM) algorithm. Our
analysis shows that the proposed method is theoretically
guaranteed to attain a reliable estimate of the true CCM with a
sample complexity at the order of $\sqrt{NM}\log{NM}$, which is
much smaller than the dimension of the CCM. Here $N$ and $M$
denote the number of antennas at the transmitter and the number of
reflecting elements at the IRS, respectively. Thus the proposed
method can help achieve a significant training overhead reduction.
Simulation results show that, with a small amount of training
overhead, the proposed method can render a reliable CCM estimate
that helps achieve near-optimal two-timescale beamforming
performance.

The rest of the paper is organized as follows. Section \ref{s2}
discusses the system model, channel model, and the motivation of
the work. Details about the downlink training and the received
signal model are presented in Section \ref{s3}. Section \ref{s4}
proposes an ADMM-based algorithm for CCM estimation. Performance
of the proposed CCM estimation method is analyzed in Section
\ref{s5}. Section \ref{s6} discusses how to perform two-timescale
beamforming based on the estimated CCM. Simulation results are
provided in Section \ref{s7}, followed by concluding remarks in
Section \ref{s8}.

\emph{Notations}: Italic letters denote scalars. Boldface
lowercase and uppercase denote the vectors and matrices,
respectively. Superscripts $\left(\cdot\right)^*$,
$\left(\cdot\right)^T$ and $\left(\cdot\right)^{H}$ denote
conjugate, transpose, and conjugate transpose, respectively.
$\mathbb{E}(\cdot)$ is the expectation operator and $j=\sqrt{-1}$.
$\text{vec}(\cdot)$ is the vectorization operation, which stacks
the columns of a matrix on top of each other. $\bm A\succcurlyeq
\bm B$ means $\bm A - \bm B$ is a positive semidefinite matrix.
The transposed Khatri-Rho, Hadamard and Kronecker product are
denoted by $\bullet$, $\circ$, and $\otimes$ respectively.
$\mathbb{CN}(\mu,\sigma^2)$ means a circularly symmetric complex
Gaussian distribution with mean $\mu$ and variance $\sigma^2$.
$\mathbb{N}(\mu,\sigma^2)$ denotes a real Gaussian distribution
with mean $\mu$ and variance $\sigma^2$. $\mathbb{C}^{N\times M}$
represents the complex space with $N\times M $ dimension.


\section{Problem Formulation}
\label{s2}
\subsection{System Model}
We consider a point-to-point IRS-aided mmWave communication system
as illustrated in Fig.~\ref{f_scenario}. An IRS is deployed to
assist data transmission from the base station (BS) to an
omnidirectional-antenna user. The BS is equipped with a uniform
linear array with $N$ antennas. The IRS is a uniform planar array
consisting of $M=M_v\times M_h$ passive reflecting elements. Each
element can independently reflect the incident signal with a
reconfigurable phase shift. Let
\begin{align*}
\bm \Psi =\text{diag}\left(
e^{j\psi_{1}},\cdots,e^{j\psi_{M}}\right)
\end{align*}
denote the reflecting coefficient matrix of the IRS, where
$\psi_{m}$ is the phase shift associated with the $m$th passive
element.

\begin{figure}[!t]
\centering
\includegraphics[width=0.48\textwidth]{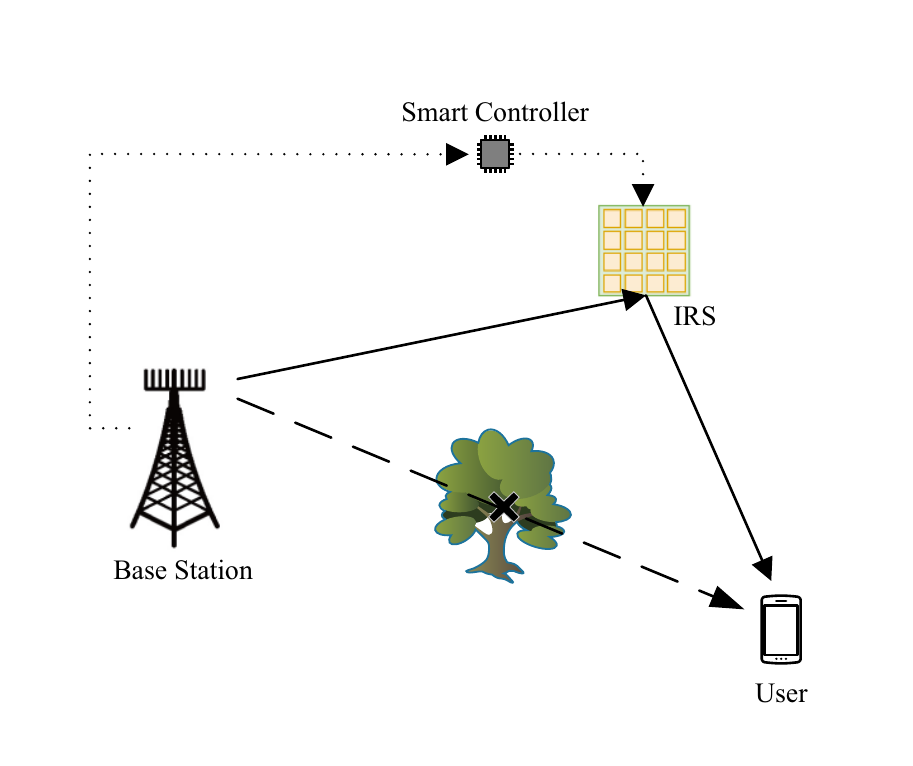}
\caption{An IRS-assisted mmWave communication system.}
\label{f_scenario}
\end{figure}



For simplicity, we assume that the direct link between the BS and
the user is blocked due to poor propagation conditions, and the
transmitted signal arrives at the user via propagating through the
BS-IRS-user channel. Let $\bm G\in\mathbb{C}^{M\times N}$ and $\bm
h\in\mathbb{C}^{M\times 1}$ denote the BS-IRS channel and the
IRS-user channel, respectively. The effective channel between the
BS and the user can thus be expressed as
\begin{align}
\bm {\tilde h}^H &= \bm h^H\bm \Psi\bm G =\bm
\psi^T\text{diag}(\bm h^H)\bm G \triangleq \bm \psi^T \bm H
\end{align}
where $\bm \psi = \text{diag}(\bm \Psi)$ and $\bm H \triangleq
\text{diag}(\bm h^H)\bm G$ is referred to as the cascade channel.
Let $s$ be the transmitted signal and $\boldsymbol{f}$ be the
transmit precoding vector. The signal received at the user can be
expressed as
\begin{align}
y & = \bm {\tilde h}^H \bm f s + n \notag\\
&= (\bm f^T\otimes \bm\psi ^T)\text{vec}(\bm H)s + n \triangleq
\bm w \text{vec}(\bm H)s + n
\end{align}
where $\bm w \triangleq \bm f^T\otimes \bm\psi ^T$ and $n$ is the
additive noise following a complex Gaussian distribution
$\mathbb{CN}(0,\sigma^2)$.


\subsection{Channel Model}
For notational convenience, we first define
\begin{align}
\bm a(\nu,D) \triangleq [1\phantom{0}  \cdots \phantom{0}
e^{j(D-1)\nu}]^T
\end{align}
We adopt a geometric mmWave channel model \cite{ElOmar14} to
characterize the channel. The BS-IRS channel $\bm G$ can be
expressed as
\begin{align}
\bm G &= \sum_{l=1}^L\alpha_l\bm
a_{r}(\theta_{v, l},\theta_{h, l})\bm a_{t}^H(\gamma_l)
\label{bs_irs}
\end{align}
where $L$ is the number of paths between the BS and the IRS,
$\alpha_l$ is the complex gain which is assumed to follow a
complex Gaussian distribution $\mathbb{CN}(0,\varpi_l^2)$,
$\{\gamma_l,\theta_{h ,l},\theta_{v ,l}\}$ are, respectively, the
angle of departure (AoD), the elevation and azimuth angle of
arrival (AoA) associated with the $l$th path; and $\bm
a_t(\gamma_l)$ and $\bm a_r(\theta_{v,l},\theta_{h,l})$ are the
transmit and receive array response vectors. Specifically, we have
\begin{align}
\bm a_t(\gamma_l) = \bm a(\nu_{1,l},N)
\end{align}
where $\nu_{1,l} \triangleq \frac{2\pi d}{\lambda}\sin (\gamma_l)$
with $\lambda$ and $d$ representing the signal wavelength and the
antenna spacing, respectively. $\bm a_r(\theta_{v, l},\theta_{h,
l})$ has a form of a Kronecker product as
\begin{align}
\bm a_r(\theta_{v, l},\theta_{h, l}) = \bm a(\nu_{2,l},M_v)\otimes \bm a(\nu_{3,l},M_h)
\end{align}
where $\nu_{2,l} \triangleq \frac{2\pi d}{\lambda}\cos
(\theta_{v,l})$~and~$\nu_{3,l} \triangleq\frac{2\pi
d}{\lambda}\sin(\theta_{v,l})\cos (\theta_{h,l})$. Define
\begin{align}
\bm\Sigma\triangleq & \text{diag}(\bm \alpha) \\
\bm\alpha \triangleq & [\alpha_1\ \cdots\ \alpha_L]^T \\
\bm F_1 \triangleq & [\bm a_r(\theta_{v ,1},\theta_{h, 1}))\phantom{0}
\cdots \phantom{0}
\bm a_r(\theta_{v, L},\theta_{h, L}) \\
\bm F_2 \triangleq & [\bm a(\nu_{1,1},N)\ \cdots\ \bm
a(\nu_{1,L},N)]
\end{align}
We can express $\bm G$ as
\begin{align}
\bm G = \bm F_1\bm \Sigma\bm F_2^H
\end{align}

Similarly, the IRS-user channel is characterized as
\begin{align}
\bm h = \sum_{p=1}^P\beta_p\bm a_r(\psi_{v, p},\psi_{h, p})
\label{irs_ue}
\end{align}
where $P$ is the number of paths between the IRS and the user,
$\beta_p$ is the complex channel gain which follows
$\mathbb{CN}(0,\chi_p^2)$, $\{\psi_{v, p},\psi_{h ,p}\}$ denotes
the elevation and azimuth AoD associated with the $p$th path, and
$\bm a_r(\psi_{v,p},\psi_{h,p})$ can be expressed as
\begin{align}
\bm a_r(\psi_{v,p},\psi_{h,p}) = \bm a(\nu_{4,p},M_v)\otimes \bm
a(\nu_{5,p},M_h)
\end{align}
where $\nu_{4,p} \!\triangleq\!\frac{2\pi d}{\lambda}\cos
(\psi_{v,p})$ and $\nu_{5,p}\!\!\triangleq\!\!\frac{2\pi
d}{\lambda}\sin(\psi_{v,p})\cos (\psi_{h,p})$. The IRS-user
channel can be written as
\begin{align}
\bm h = \bm F_3\bm \beta
\end{align}
where
\begin{align}
\bm F_3 \triangleq & [\bm a_r(\psi_{v, 1},\psi_{h, 1})\ \cdots \
\bm a_r(\psi_{v ,P},\psi_{h, P})] \\
\bm \beta\triangleq & [\beta_1\ \cdots\ \beta_P]^T
\end{align}


\subsection{Motivations}
Although I-CSI helps achieve optimal beamforming performance,
system optimization based on I-CSI is considered as an
intimidating task due to the associated high computational
complexity and training overhead as noted in Section I.


To address this challenge, notice that acquiring the I-CSI of the
equivalent channel between the BS and the user is much easier and
can be accomplished by conventional channel estimation schemes.
Inspired by this fact, it is natural to consider a two-timescale
beamforming protocol~\cite{ZhaoWu21}. Specifically, in the
long-term time-scale, the reflecting coefficients at the IRS are
designed according to the statistical CSI which varies much more
slowly than the I-CSI. Then, given fixed passive reflecting
coefficients, the BS's transmit beamforming matrix can be
optimally determined based on the instantaneous effective channel
in a short-term time-scale.




A prerequisite for such a two-timescale beamforming protocol is to
obtain the statistical CSI, i.e. the spatial CCM of the cascade
channel. Nevertheless, how to efficiently obtain an estimate of
the CCM for IRS-assisted systems has not been fully considered
before. The objective of this work is to present a method to
estimate the CCM for IRS-aided mmWave systems. Specifically, by
exploiting the PSD multi-level Toeplitz and low-rank structure of
the CCM, we develop a method which is theoretically guaranteed to
attain a reliable estimate of the true CCM with a sample
complexity that is much smaller than the dimension of the CCM.






\section{Downlink Training}
\label{s3} In order to estimate the CCM, we employ a downlink
training procedure consisting of $T$ time frames. We assume that
each time frame has a short period of time so that the channel
remains unaltered. Each time frame consists of multiple training
time slots, say, $J$ time slots. During the training process, the
BS sends a same pilot signal, i.e., $s_{t,j} = 1,\ \forall
j\in\{1,\cdots,J\}$, to the receiver. The pilot signal is precoded
via a transmit precoding vector $\bm f_{t,j}$ that changes over
different time frames and different time slots. The signal
received at the user can be expressed as
\begin{align}
y_{t,j} & = \bm h_t^H\bm \Psi_{t,j}\bm G_t\bm f_{t,j} + n_{t,j}\notag\\
&= \bm w_{t,j}\text{vec}\left(\bm H_t\right) + n_{t,j}
\end{align}
where $\boldsymbol{h}_t$ and $\boldsymbol{G}_t$ respectively
represent the IRS-user channel and the BS-IRS channel at the $t$th
time frame, $\bm \Psi_{t,j}$ is the phase shift matrix that is
employed at the $j$th time slot of the $t$th time frame, $n_{t,j}$
is the additive white Gaussian noise, $\bm H_t\triangleq
\text{diag}(\bm h_t^H)\bm G_t$ is referred to as the cascade
channel at the $t$th frame, and $\bm w_{t,j} \triangleq \bm
f_{t,j}^T\otimes \bm \psi_{t,j}^T$, in which $\bm \psi_{t,j} =
\text{diag}(\bm \Psi_{t,j})$. Define $\bm y_t \triangleq
[y_{t,1}\cdots y_{t,J}]^T$, $\bm n_{t}\triangleq[n_{t,1}\cdots
n_{t,J}]^T$ and $\bm W_t \triangleq [\bm w_{t,1}^T \ \cdots \ \bm
w_{t,J}^T]^T$. To facilitate CCM estimation, we employ the same
$\bm W = \bm W_t,\forall t$ for each time frame. Thus the received
signal at the $t$th frame can be written in a matrix form as
\begin{align}
\bm y_t = \bm W\text{vec}\left(\bm H_t\right) + \bm n_t =\bm W
\boldsymbol{\bar{h}}_t + \bm n_t \label{y_i}
\end{align}
where $\boldsymbol{\bar{h}}_t\triangleq\text{vec}\left(\bm
H_t\right)$.



As reported by previous studies
\cite{VieringHofstetter02,AlkhateebEl14}, an important fact about
mmWave channels is that the angle parameters such as the AoA and
AoD depend only on the relative positions of the BS, the user, and
the scatterers, which vary much more slowly than I-CSI. Also, the
statistics of the complex path gain keep invariant over an
interval that is much longer than the channel coherent time. Hence
it is reasonable to make the following basic assumption:
\begin{enumerate}
\item[A1] The angle parameters and the
variances of the path gains, remains unchanged over a long-term
period.
\end{enumerate}
Our objective is to obtain an estimate of
$\boldsymbol{R}_{h}\triangleq
\mathbb{E}[\boldsymbol{\bar{h}}_t\boldsymbol{\bar{h}}_t^H]$ from
the received signal $\{\bm y_t\}_{t=1}^T$. Note that the CCM
$\boldsymbol{R}_{h}$ has a dimension of $NM\times NM$. We will
show that the knowledge of $\boldsymbol{R}_{h}$ suffices for
optimizing the reflecting coefficients at the IRS.



\subsection{Discussions}
Although this paper considers CCM estimation of the downlink
channel, the proposed method can be readily applied to the
uplink's CCM estimation. Similar to the downlink training, the
uplink training contains $T$ time frames and each time frame
consists of $J$ time slots. During the training phase, the user
sends a same pilot signal $s_{t,j}=1$ to the BS, where the
received signal is combined via a combing vector $\bm g_{t,j}$.
The channel in the same time frame is assumed to be
time-invariant. Therefore the signal received at the $j$th time
slot of the $t$th time frame is given as
\begin{align}
{y}_{t,j}^u& =\bm g_{t,j}^H \bm G_{t}^H \bm \Psi_{t,j} \bm h_t s_{t,j} +{ n}_{t,j}^u\notag\\
& = {\bm w}_{t,j}^u \text{vec}(\bm {{ H}}_t^u) +  {
n}_{t,j}^u \label{up_model}
\end{align}
where $\bm G_t$ and $\bm h_t$ are, respectively, the IRS-BS
channel and the user-IRS channel, ${\bm w}_{t,j}^u\triangleq \bm
\psi_{t,j}^T\otimes \bm g_{t,j}^H$ with $\bm \psi_{t,j} =
\text{diag}(\bm \Psi_{t,j})$ and $\bm { H}_t^u = \bm G_t^H
\text{diag}(\bm h_t)$. Define $\bm { y}_t^u \triangleq [y_{t,1}^u\
\cdots\ y_{t,J}^u ]^T$, $\bm {n}_t^u \triangleq [n_{t,1}^u\
\cdots\ n_{t,J}^u ]^T$, and $\bm { W}_t^u \triangleq [(\bm{
w}_{t,1}^u)^T \ \cdots \ (\bm{ w}_{t,J}^u)^T]$. The received
signal at the BS can be expressed as
\begin{align}
\bm { y}_t^u = \bm{{W}}_t^u\text{vec}(\bm{ H}_t^u) + \bm {
n}_t^u \label{up_train}
\end{align}
If we set $\bm{{W}}_t^u$ the same for different time frames, the
signal model for the uplink training has a same form as that for
the downlink scenario. For time division duplex (TDD) systems, due
to the channel reciprocity between opposite links (downlink and
uplink), the estimated uplink CCM can be used for downlink
precoding/beamforming.




\section{Spatial Channel Covariance Matrix Estimation}
\label{s4}
According
to~\eqref{y_i}, we have
\begin{align}
\bm R_y  = \mathbb{E}\{\bm y_t\bm y_t^H\} = \bm W\bm R_h\bm W^H +
\sigma^2\bm I
\end{align}
Generally the true covariance matrix $\bm R_y$ is unavailable. But
it can be estimated via the following sample covariance matrix
\begin{align}
\bm{\hat{ R}}_y= \frac{1}{T} \sum_{t=1}^T\bm y_t\bm y_t^H
\end{align}
Intuitively, one can directly estimate $\bm R_h$ from $\bm{\hat
R_y}$ via solving the following least squares problem
\begin{align}
\text{vec}(\bm {\hat R_y}) = \left(\bm W^*\otimes \bm W
\right)\text{vec}({\bm R_h}) + \text{vec}(\sigma^2\bm I)
\end{align}
provided that the dimension of $\bm{\hat R_y}$ (i.e. $J\times J$)
is larger than the dimension of $\bm R_h$ (i.e. $MN\times MN$).
Nevertheless, the condition $J\geq NM$ is unlikely to be satisfied
in practice since the coherence time in mmWave systems is
relatively small. Therefore, estimating $\bm R_h$ from $\bm{\hat
R_y}$ is in fact an underdetermined problem, and in order to
handle such an issue we have to exploit the structure of $\bm
R_h$.

\subsection{Exploiting The Structure of $\bm R_h$}
To exploit the structure of $\bm R_h$, we first obtain a sparse
representation of the cascade channel $\boldsymbol{H}_t$. For
notational convenience, in the following we will omit the
subscript $t$ in $\boldsymbol{H}_t$, $\boldsymbol{h}_t$ and
$\boldsymbol{G}_t$. Utilizing the matrix properties, the cascade
channel $\bm H$ can be expressed as
\begin{align}
\bm H & = \text{diag}(\bm h^H)\bm G=\bm h^* \bullet \bm G\notag\\
& =(\bm F_3^*\bm \beta^*)\bullet (\bm F_1\bm \Sigma\bm F_2^H)\notag\\
& = \left(\bm F_3^*\bullet\bm F_1\right)\left(\bm \beta^*\otimes \left(\bm \Sigma\bm F_2^H\right)\right)\notag\\
& = \left(\bm F_3^*\bullet\bm F_1\right)\left(\bm \beta^*\otimes\bm \Sigma\right)\bm F_2^H\notag\\
&\triangleq \bm F_4 \bm \Pi \bm F_2^H \label{cascade_H}
\end{align}
where
\begin{align}
\bm F_4\triangleq&\bm F_3^*\bullet\bm F_1 \\
\bm \Pi\triangleq & \bm \beta^*\otimes\bm \Sigma
\end{align}
Note that the $\varsigma$th ($\varsigma = (p-1)L+l$) column of
$\bm F_4\in\mathbb{C}^{M\times PL}$ is given by
\begin{align}
&\bm a_r^*(\psi_{v, p},\psi_{h, p})\circ \bm a_r(\theta_{v ,l},\theta_{h, l})\notag\\
&\overset{(a)}{=}\left(\bm a^*(\nu_{4,p},M_v)\circ \bm a(\nu_{2,l},M_v)\right)\notag\\
&\qquad\quad\otimes\left(\bm a^*(\nu_{5,p},M_h)\circ \bm a(\nu_{3,l},M_h)\right)\notag\\
&\triangleq \bm a(\nu_{6,\varsigma},M_v)\otimes \bm a(\nu_{7,\varsigma},M_h)
\end{align}
where $\nu_{6,\varsigma} \triangleq\nu_{2,l} -\nu_{4,p} $,
$\nu_{7,\varsigma}\triangleq\nu_{3,l} -\nu_{5,p}$, and $(a)$ follows from
the property: $(\bm A\otimes \bm B)\circ (\bm C\otimes\bm D) =
(\bm A\circ \bm C)\otimes (\bm B\circ\bm D)$.


Vectorizing the cascade channel in~\eqref{cascade_H} leads to
\begin{align}
\text{vec}(\bm H) = \text{vec}(\bm F_4 \bm \Pi \bm F_2^H)
\triangleq \bm F\bm x \label{vec_H}
\end{align}
where $\bm F \triangleq \bm F_2^*\otimes \bm
F_4\in\mathbb{C}^{MN\times L^{2}P}$ and $\bm x \triangleq
\text{vec}(\bm\Pi)\in\mathbb{C}^{L^2 P}$. $\bm F$ can be further
expressed as
\begin{align}
\bm F &= \bm F_2^*\otimes \bm F_4 \notag\\
&= [\bm a^*(\nu_{1,1},N)\ \cdots\ \bm a^*(\nu_{1,L},N)] \notag\\
&\qquad\ \otimes \Big[\bm a(\nu_{6,1},M_v)\otimes \bm a(\nu_{7,1},M_h)\ \cdots \notag\\
&\qquad\qquad\qquad\qquad \bm a(\nu_{6,LP},M_v)\otimes \bm
a(\nu_{7,LP},M_h)\Big] \label{F}
\end{align}
Substituting~\eqref{F} into~\eqref{vec_H}, we have
\begin{align}
&\bm{\bar h}=\text{vec}(\bm H)\notag\\
 &= \sum_{l=1}^{L}\sum_{\zeta=1}^{LP}
x_{\varrho}\bm a^*(\nu_{1,l},N)\otimes\bm
a(\nu_{6,\zeta},M_v)\otimes\bm a(\nu_{7,\zeta},M_h) \label{vvH}
\end{align}
where $x_{\varrho}$ with $\varrho = \zeta + (l-1)LP$ is the
$\varrho$th element of $\bm x$. Due to the fact that $\bm
\Pi=\bm\beta^*\otimes \bm\Sigma$ and $\bm \Sigma $ is a diagonal
matrix, there are at most $LP$ non-zeros elements in $\bm x$, and
each element in $\bm x$ is given by
\begin{align}
x_{\varrho} = \left\{
\begin{array}{ll}
\alpha_l\beta_p^*,&\varrho\in\bm O\\
0,&\text{otherwise}\\
\end{array}
\right.
\end{align}
where the set $\bm O$ is defined as
\begin{align*}
\Big\{(p-1)L+l +
(l-1)LP|l\in\{1,\cdots,L\};p\in\{1,\cdots,P\}\Big\}
\end{align*}

Recall that $\alpha_l\sim\mathbb{CN}(0,\varpi_l^2)$ and
$\beta_p\sim\mathbb{CN}(0,\chi_p^2)$, and $\alpha_l$ and $\beta_p$
are mutually uncorrelated. Therefore for
$x_{\varrho}=\alpha_l\beta_p^*$, its mean and variance are given
as
\begin{align}
\mathbb{E}[x_{\varrho}] & = 0 \notag\\
\mathbb{E}[ x_{\varrho} x_{\varrho}^*] & = \eta_{\varrho}^2
\triangleq \varpi_l^2\chi_p^2
\end{align}
We see that $\bm{\bar h}$ can be characterized by a geometric
channel model. It has $LP$ uncorrelated composite paths in total.
The complex gain of each composite path is a random variable with
zero mean and finite variance. The angular parameters
($\{\nu_{1,l},\nu_{6,\zeta},\nu_{7,\zeta}\}$) associated with each
path are treated as deterministic parameters as angle parameters
vary slowly relative to the complex path gains. Hence the channel
covariance matrix $\bm R_h$ can be expressed as
\begin{align}
\bm R_h & = \mathbb{E}(\bm{\bar h}\bm {\bar h}^H) = \bm F\mathbb{E}(\bm x\bm x^H)\bm F^H\notag\\
& \overset{(a)}{=} \sum_{{\varrho}=1}^{L^2P} \mathbb{E}( x_{\varrho} x_{\varrho}^*)\bm R_{\varrho}\notag\\
& = \sum_{{\varrho}\in \bm O} \eta_{\varrho}^2\bm R_{\varrho}
\label{Rh}
\end{align}
where $(a)$ follows from the fact that $\mathbb{E}(\bm x\bm x^H)$
is a diagonal matrix, and $\bm
R_{\varrho}\in\mathbb{C}^{{NM}\times{NM}}$ is defined as
\begin{align}
\bm R_{\varrho} &= \left(\bm a^*(\nu_{1,l},N)\bm a^T(\nu_{1,l},N)\right)\notag\\
&\qquad\otimes \left(\bm a(\nu_{6,\zeta},M_v)\bm a^H(\nu_{6,\zeta},M_v)\right)\notag\\
&\qquad\qquad\otimes\left(\bm a(\nu_{7,\zeta},M_h)\bm
a^H(\nu_{7,\zeta},M_h)\right)
\end{align}
It can be easily verified that $\bm R_{\varrho}$ is a PSD 3-level
Toeplitz matrix. As a result, $\bm
R_h\in\mathbb{C}^{{NM}\times{NM}}$ is also a PSD 3-level Toeplitz
matrix. Although there are $N^2 M^2$ elements in $\bm R_h$, owing
to the specific structure of PSD 3-level Toeplitz matrix, $\bm R_h$
can be characterized by $(2N-1)(2M_v-1)(2M_h-1)$ parameters which
can be represented by a third-order tensor $\bm
V\in\mathbb{C}^{(2N-1)\times(2M_v-1)\times(2M_h-1)}$, i.e. we can
write $\bm R_h=\mathbb{T}_3(\bm V)$. How to map a third-order
tensor to a 3-level Toeplitz can be found in Appendix~\ref{ap4}.
Furthermore, from~\eqref{Rh}, we know that $\bm R_h$ can be
represented by a summation of $LP$ rank-one matrices. Due to the
sparse scattering characteristics of mmWave channels, $LP$ is
usually much smaller than the dimension of $\bm R_h$ (i.e., $NM$),
meaning that $\bm R_h$ has a low-rank structure.

Utilizing the PSD 3-level Toeplitz structure and the low-rank property
of $\mathbb{T}_3(\bm V)$, the estimation of $\bm R_h$ can be cast
into the following low-rank structured covariance reconstruction
problem:
\begin{align}
\bm{\hat R}_h = &\arg\min_{\mathbb{T}_3(\bm V)}\ \frac{1}{2}
\left\|\boldsymbol{\hat{R}}_y -
\bm W\mathbb{T}_3(\bm V)\bm W^H\right\|_F^2 +\lambda \text{rank}(\mathbb{T}_3(\bm V))\notag\\
& \qquad\text{s.t.}\quad \mathbb{T}_3(\bm V)\succcurlyeq 0
\label{rank_min}
\end{align}
where $\lambda$ is a regularization parameter to balance the
tradeoff between data fitting and low-rankness. Nonetheless, such
a problem is generally NP-hard due to the rank function. To make
it tractable, we resort to convex relaxation to replace
$\text{rank}(\mathbb{T}_3(\bm V))$ with the nuclear-norm of
$\mathbb{T}_3(\bm V)$. Since $\mathbb{T}_3(\bm V)$ is confined to
be a PSD matrix, its nuclear norm is equivalent
to its trace. Consequently, the resulting optimization can be
given by
\begin{align}
\bm{\hat R}_h = &\arg\min_{\mathbb{T}_3(\bm V)}\quad \frac{1}{2}
\left\|\boldsymbol{\hat{R}}_y -
\bm W\mathbb{T}_3(\bm V)\bm W^H\right\|_F^2 +\lambda \text{tr}(\mathbb{T}_3(\bm V))\notag\\
& \qquad\text{s.t.}\qquad \mathbb{T}_3(\bm V)\succcurlyeq 0
\label{sdp}
\end{align}
The above optimization is a convex semidefinite programming (SDP)
problem which can be solved by many standard off-the-shelf
solvers, e.g., CVX. Unfortunately, these solvers are usually
computationally expensive. To reduce the computational complexity,
we develop an alternating direction method of multipliers (ADMM)
algorithm for solving ~\eqref{sdp} in the next section.


\subsection{ADMM-Based Algorithm}
To solve~\eqref{sdp}, we first introduce two auxiliary variables,
$\bm A$ and $\bm B$, and reformulate \eqref{sdp} into the
following optimization
\begin{align}
\{\bm{\hat V},\bm{\hat{ A}},\bm{\hat{B}}\} =&\arg\min_{\bm V,\bm A,\bm B}\ \Big(\frac{1}{2}
\left\|\boldsymbol{\hat{R}}_y - \bm W\bm A\bm W^H\right\|_F^2 \notag\\
&\qquad\qquad\qquad +\lambda \text{tr}(\bm A) + \mathbb{I}_{\infty}(\bm B\succcurlyeq 0)\Big) \notag\\
& \qquad\text{s.t.}\qquad \bm A = \mathbb{T}_3(\bm V),\ \bm B =
\bm A,\
\end{align}
where $\mathbb{I}_{\infty}(a)$ is an indicator function defined as
\begin{align}
\mathbb{I}_{\infty}(a) = \left\{\begin{array}{ll}
0,&\quad \text{if $a$ is true}\\
\infty,&\quad\text{otherwise}
\end{array}\right.
\end{align}
The augmented Lagrangian of the above optimization is read as
\begin{align}
&\mathcal{L}(\bm V,\bm A,\bm B,\bm\Upsilon,\bm \Lambda)\notag\\
&= \frac{1}{2}\left\|\bm{\hat {R}}_y - \bm W  \bm A \bm W^H\right\|_F^2 + \lambda \text{tr}(\bm A) \notag\\
&\qquad+ \left\langle \bm\Upsilon,\bm A - \mathbb{T}_3(\bm V) \right\rangle + \frac{\eta}{2}
\left\|\bm A - \mathbb{T}_3(\bm V)\right\|_F^2\notag\\
&\qquad\quad +  \left\langle \bm\Lambda,\bm B - \bm A
\right\rangle + \frac{\rho}{2}\left\|\bm B - \bm
A\right\|_F^2+\mathbb{I}_{\infty}(\bm B\succcurlyeq 0)
\end{align}
where $\left\langle \bm A,\bm B \right\rangle$ is defined as
$\text{Re}\big(\text{Tr}(\bm B^H\bm A)\big)$, $\bm \Upsilon$ and
$\bm \Lambda$ are the dual parameters, and $\eta,\rho > 0$ are the
penalty parameters. According to the updating rule of the ADMM
algorithm, it consists of solving the following sub-problems:
\begin{align}
\bm {\hat{A}}^{k+1}&= \arg\min_{\bm A}\
\mathcal{L}(\bm V^k,\bm A,\bm B^k,\bm\Upsilon^k,\bm \Lambda^k)\label{A}\\
\bm {\hat{V}}^{k+1}&= \arg\min_{\bm V}\
\mathcal{L}(\bm V,\bm A^{k+1},\bm B^k,\bm\Upsilon^k,\bm \Lambda^k)\label{V}\\
\bm {\hat{B}}^{k+1}&= \arg\min_{\bm B}\
\mathcal{L}(\bm V^{k+1},\bm A^{k+1},\bm B,\bm\Upsilon^k,\bm \Lambda^k)\label{B}\\
\bm \Upsilon^{k+1} & = \bm \Upsilon^k + \eta\big(\bm A^{k+1} - \mathbb{T}_3(\bm V^{k+1})\big)\\
\bm \Lambda^{k+1} & = \bm \Lambda^{k} + \rho \big (\bm B^{k+1} -
\bm A^{k+1}\big)
\end{align}

\emph{Update of $\boldsymbol{A}$:} We first solve the sub-problem
\eqref{A}. Calculating the derivative of $\mathcal{L}(\bm V^k,\bm
A,\bm B^k,\bm\Upsilon^k,\bm \Lambda^k)$ with respect to $\bm A$,
we have
\begin{align}
\nabla_{\bm A} \mathcal{L} &= (\bm W^H\bm W)\bm A (\bm W^H\bm W)^H + (\eta-\rho)\bm A-
\eta \mathbb{T}_3(\bm V^{k}) \notag\\
& \qquad + \rho \bm B^k -\bm W^H\bm
{\hat R}_y\bm W  + \bm\Upsilon^k - \bm \Lambda^k+\lambda\bm  I
\end{align}
By setting $\nabla_{\bm A} \mathcal{L}$ to $0$, we can update $\bm
A$ via solving the following linear equation
\begin{align}
\bm \Xi\bm A \bm \Xi^H + \kappa \bm A = \bm C \label{A_solved}
\end{align}
where $\bm C = \eta \mathbb{T}_3(\bm V^{k}) - \rho \bm B^k + \bm
W^H\bm {\hat R}_y\bm W-\lambda\bm  I- \bm\Upsilon^k + \bm \Lambda^k$, $\kappa = \eta-\rho$ and
$\bm\Xi = \bm W^H\bm W$. Clearly, $\bm A$ can be solved via
\begin{align}
(\bm \Xi^*\otimes \bm \Xi + \kappa \bm I)\text{vec}(\bm A) =
\text{vec}(\bm C)
\end{align}
Nevertheless, this approach has a computational complexity of
$\mathcal{O}(N^6M^6)$. To reduce the computational complexity, we
propose a more computationally-efficient approach to
solve~\eqref{A_solved}. Since $\bm \Xi = \bm W^H\bm W$, $\bm \Xi$
can be diagonalized via the eigen-decomposition (EVD), i.e.,
\begin{align}
\bm \Xi = \bm A_1\bm A_2 \bm A_1^{H} \label{relation1}
\end{align}
where $\bm A_2$ is a diagonal matrix and $\bm A_1$ is a unitary
matrix (meaning $\bm A_1^{H} = \bm A_1^{-1}$).
Substituting~\eqref{relation1} into~\eqref{A_solved} results in
\begin{align}
\bm A_1\bm A_2 \bm A_1^{H}\bm A(\bm A_1\bm A_2 \bm A_1^{H})^H +
\kappa \bm A = \bm C \label{relation2}
\end{align}
By defining $\bm A_3 \triangleq \bm A_1^{H}\bm A \bm A_1$ and $\bm
C_1 \triangleq \bm A_1^{H}\bm C \bm A_1$, we can express
\eqref{relation2} as
\begin{align}
\bm A_2 \bm A_3 \bm A_2 + \kappa \bm A_3 = \bm C_1
\label{relation3}
\end{align}
Since $\bm A_2$ is a diagonal matrix, we can explicitly
solve~\eqref{relation3} in an elementwise manner. Once $\bm A_3$
is obtained, $\bm A$ can be simply reconstructed as
\begin{align}
\bm A = \bm A_1\bm A_3 \bm A_1^H
\end{align}


\emph{Update of $\boldsymbol{V}$:} Keeping the terms that only
depend on $\bm V$ in~\eqref{V}, we have
\begin{align}
\bm {\hat{V}}^{k+1}=&\arg\min_{\bm V}\ \frac{\eta}{2} \text{Tr}
\left(\mathbb{T}_3^H(\bm V)\mathbb{T}_3(\bm V)\right) \notag\\
& \qquad\quad\  -\text{Re}\left(\text{Tr}\big((\bm\Upsilon^{k}+\eta\bm
A^{k+1})^H\mathbb{T}_3(\bm V)\big)\right) \label{cost_V}
\end{align}
Directly taking derivative with respect to $\boldsymbol{V}$ is
difficult. Nevertheless, the elements in $\bm V$ can be optimized
separately. To this end, we reshape $\bm
V\in\mathbb{C}^{(2N-1)\times(2M_v-1)\times(2M_h-1)}$ to a vector
$\bm v\in\mathbb{C}^{(2N-1)(2M_v-1)(2M_h-1)}$. Also, we define an
index set $\mathbb{I}_{v_l}$ as
\begin{align}
\mathbb{I}_{v_l} &= \{ (i,j) | \mathbb{T}_3(\bm V)(i,j)  \equiv
v_l \}
\end{align}
for all $l\in \{1,\cdots,(2N-1)(2M_v-1)(2M_h-1)\}$. The
cardinality of the set $\mathbb{I}_{v_l}$ is denoted by
$|\mathbb{I}_{v_l}|$. Based on this definition, the terms that
only depend on $v_l$ in the cost function~\eqref{cost_V} can be
rewritten as
\begin{align}
\hat{v}_l^{k+1} = \arg\min_{v_l}\frac{\eta|\mathbb{I}_{v_l}|}{2}
v_lv_l^* - \bigg(\sum_{(i,j)\in \mathbb{I}_{v_l}}
\text{Re}\left(\Delta^*(i,j)v_l\right)\bigg)
\end{align}
where $\Delta\triangleq \bm\Upsilon^{k} + \eta\bm A^{k+1}$. It is
clear that $\hat{v}_l^{k+1}$ is given by
\begin{align}
\hat{v}_l^{k+1} =
\frac{1}{\eta|\mathbb{I}_{v_l}|}\bigg(\sum_{(i,j)\in
\mathbb{I}_{v_l}} \Delta(i,j)\bigg)
\end{align}

\emph{Update of $\boldsymbol{B}$:} The optimization in~\eqref{B}
is equivalent to
\begin{align}
\bm {\hat{B}}^{k+1}= \arg\min_{\bm B}\ &\frac{\rho}{2}\text{Tr}(\bm B\bm B^H) +
\text{Re}(\text{Tr}((\bm\Lambda^{k}-\rho\bm A^{k+1})^H\bm B))\notag\\
&\qquad+\mathbb{I}_{\infty}(\bm B\succcurlyeq 0) \label{up_B}
\end{align}
The last term in~\eqref{up_B} makes the optimization problem
intractable. To handle this issue, we propose a two-step solution.
Specifically, we first ignore the positive semi-definite
constraint and solve the following optimization problem
\begin{align}
\bm {\tilde{B}}= \arg\min_{\bm B}\
\frac{\rho}{2}\text{Tr}(\bm B\bm B^H) +
\text{Re}(\text{Tr}((\bm\Lambda^k-\rho\bm A^{k+1})^H\bm B))
\end{align}
which gives
\begin{align}
\bm {\tilde{B}} =\bm A^{k+1}-\bm\Lambda^k/\rho
\end{align}
Then we calculate $\bm {\hat{B}}^{k+1}$ by projecting $\bm
{\tilde{B}}$ onto the positive semi-definite cone, which is
equivalent to setting all negative eigenvalues of $\bm
{\tilde{B}}$ to zero.


\section{Performance Analysis of The CCM Estimator}
\label{s5}
In this section, we analyze the estimation performance of the CCM
estimator (\ref{sdp}). Specifically, we are interested in
quantifying the amount of training overhead required to achieve a
reliable estimate of the true CCM $\bm R_h$. To facilitate our
analysis, we consider the noise-free case, i.e., $\sigma^2 = 0$
and $\bm R_y = \bm W \bm R_h \bm W^H$. In the following we use
$\bm R_h$ and $\mathbb{T}_3(\bm V)$ interchangeably since these
two essentially have the same meaning.

Suppose $\mathbb{T}_3(\bm X)\!\in\!\mathbb{C}^{I_1 I_2 I_3\times
I_1 I_2 I_3}$ is a 3-level Toeplitz matrix parameterized by a
tensor $\bm X \in\mathbb{C}^{(2I_1-1)\times (2I_2-1)\times (2I_3 -
1)}$. For any matrix $\bm M\in\mathbb{C}^{I\times I_1I_2I_3}$,
$\bm {\check M}\in\mathbb{C}^{I^2\times(2I_1-1)(2I_2-1)(2I_3 -
1)}$ is referred to as the transforming matrix of $\bm M$ if it
satisfies
\begin{align}
\bm {\check M} \bm x = \text{vec}(\bm M \mathbb{T}_3(\bm X) \bm
M^H)
\end{align}
where $\bm x$ is a vector by reshaping $\bm X$ into a vector. In
addition, define $r_{e}(\bm A)\triangleq \frac{\text{Tr}(\bm
A)}{\|\bm A\|_2}$ as the effective rank of matrix $\bm A$. Our
result is summarized as follows.

\begin{theorem}
\label{t1} Let $\bm V\in\mathbb{C}^{(2N-1)\times
(2M_v-1)\times(2M_h-1)}$ be the ground truth and $\bm{\hat
V}\in\mathbb{C}^{(2N-1)\times (2M_v-1)\times(2M_h-1)}$ be the
solution of~\eqref{sdp}. Given observations $\{\bm y_t\}_{t=1}^T$,
and set
\begin{align}
J\ge u\triangleq\sqrt{(2N-1)(2M_v-1)(2M_h-1)} \label{cond_J}
\end{align}
and
\begin{align}
\lambda \ge c \|\bm W\|_F^2 \|\bm R_y \|_2\text{max}\{\sqrt{
\tilde{\delta}}, \tilde{\delta}\} \label{cond_lambda}
\end{align}
where $c$ is a constant and $ \tilde{\delta}$ is defined as
\begin{align}
\tilde{\delta} \triangleq \frac{r_{e}(\bm R_y)\log(TJ)}{T}
\label{sigma1}
\end{align}
then with probability at least $1-4T^{-1}$, the average per-entry
root mean square error (RMSE) of the solution to~\eqref{sdp}
satisfies
\begin{align}
\frac{1}{u}\|\bm{\hat V} - \bm V\|_F\le
\frac{16\lambda\sqrt{r}}{\sigma^2_{\text{min}}(\bm {\check W})}
\frac{\sqrt{NM}}{u} \label{v_error}
\end{align}
where $r$ is the rank of $\bm R_h$, $\bm {\check W}$ is the
transforming matrix of $\bm W$, and $\sigma_{\text{min}}(\bm
{\check W})$ denotes the smallest singular value of $\bm {\check
W}$.
\end{theorem}
\begin{proof}
See Appendix \ref{A1}.
\end{proof}



The above theorem is a generalization of Theorem 4
in~\cite{LiChi15}. Specifically, \cite{LiChi15} analyzes the
structured covariance estimation performance under a partial
observation framework where the aim is to recover the complete CCM
from a submatrix of the sample CCM. Theorem~\ref{t1} generalizes
this partial observation model to an arbitrarily linear
compression model in which we do not have direct access to the
entries of the sample CCM; instead, only the sample covariance
matrix $\boldsymbol{\hat{R}}_y$ is available. Also, the CCM in
this work has a multi-level Toeplitz structure, which is different
from that of~\cite{LiChi15}.

Recalling $M=M_v M_h$, the term $\sqrt{NM}/u$ in \eqref{v_error}
tends to be a constant for sufficiently large values of $M$ and
$N$. Therefore, the average per-entry RMSE is upper bounded by
$\lambda\sqrt{r}$ times a scale factor. According to
\eqref{cond_lambda}, we know that
\begin{align}
\lambda\sqrt{r}\geq \bar u \triangleq c \sqrt{r} \|\bm W\|_F^2
\|\bm R_y \|_2\text{max}\{\sqrt{  \tilde{\delta}},
\tilde{\delta}\} \label{lambda_r}
\end{align}
Therefore, in our optimization problem~\eqref{sdp}, we set
$\lambda=\bar u/\sqrt{r}$ such that the average pre-entry RMSE of
$\bm V$ has the smallest upper bound, i.e.,
\begin{align}
&\frac{1}{u}\|\bm {\hat V} - \bm V\|_F\notag\\
&\le\frac{\sqrt{NM}}{u}\frac{16 \bar u}{\sigma^2_{\text{min}}(\bm{\check W})}\notag\\
&=\frac{\sqrt{NM}}{u}\frac{16 c \|\bm W\|_F^2 \|\bm R_y
\|_2\text{max}\{\sqrt{  r\tilde{\delta}}, \sqrt{r}\tilde{\delta}\}}{\sigma^2_{\text{min}}(\bm{\check W})}
\label{upper_bound}
\end{align}
From~\eqref{upper_bound}, we can see that the average pre-entry
RMSE vanishes as long as $r\tilde{\delta}$ tends to $0$. To this
objective, it can be verified that the number of time frames $T$
should be in the order of $r_e(\bm R_y)r\log(J)$ or in the order
of $r^2\log (J)$ since the effective rank of a matrix is no
greater than its true rank, i.e. $r_e(\bm R_y)\leq
r$~\cite{NegahbanRavikumar12}. To see why $r\tilde{\delta}$ tends
to $0$ when $T\sim \mathcal{O}(r^2\log (J))$, let $T = \tau r^2
\log (J)$, where $\tau$ is a constant. In this case, we have
\begin{align}
r\tilde{\delta}\leq\frac{\log(\tau r^2\log(J))}{\tau\log(J)}
\end{align}
where the right-hand side of the above inequality decreases to a
small value as $\tau$ increases. In summary, when the number of
time frames $T$ is in the order of $r^2\log (J)$, the average
per-entry RMSE can be upper bounded by an arbitrarily small value,
which means that we can obtain a reliable estimate of the true CCM
$\boldsymbol{R}_h$. Note that our performance guarantee is
non-asymptotic and holds for a finite number of measurement
vectors. In other words, the result has accounted for the
covariance estimation error due to finite samples.


Recall that the proposed downlink training protocol consists of
$T$ time frames, and each time frame comprises $J$ time slots.
Thus the total amount of training overhead is $TJ$. Since the
number of time slots $J$ has to satisfy \eqref{cond_J}, which
means that $J$ is on the order of $\sqrt{NM}$. Therefore, the
total amount of training overhead is at the order of
$r^2\sqrt{NM}\log(NM)$. Note that here $r = PL$ is the rank of
$\bm R_h$. In mmWave systems, due to sparse scattering
characteristics, both $P$ and $L$ are relatively small. Hence $r$
is generally far less than $NM$. Based on this, we can conclude
that when the total number of training symbols is in the order of
$\sqrt{NM}\log{NM}$, we can provide a reliable estimate of the CCM
$\boldsymbol{R}_h$. Note that $\sqrt{NM}\log{NM}$ is much smaller
than the dimension of $\bm R_h$ (i.e., $N^2 M^2$).

\section{Two-Timescale Beamforming Based on Estimated CCM}
\label{s6}
In this section, we discuss how to perform two-timescale
beamforming based on the estimated CCM. For the considered
point-to-point IRS-assisted mmWave system, the received signal at
the user can be expressed as
\begin{align}
y_t = \bm h^H\bm \Psi\bm G\bm f s_t + n_t
\end{align}
where $s_t $ is the transmitted symbol satisfying
$\mathbb{E}(|s_t|^2) = 1$, $\bm f$ is the precoding vector, and
$n_t\sim\mathbb{CN}(0,\sigma^2)$ denotes the zero-mean complex
Gaussian noise. The achievable spectral efficiency can be
expressed as
\begin{align}
R = \log_2(1 + {|\bm h^H\bm \Psi\bm G \bm f|^2}/{\sigma^2})
\end{align}
To circumvent the need for the I-CSI $\boldsymbol{h}$ and
$\boldsymbol{G}$, we adopt a two-timescale beamforming approach.
Specifically, in the long-term time-scale, the reflecting
coefficients at the IRS are optimized based on the estimated CCM.
Then, given fixed passive reflecting coefficients, the BS's
transmit beamforming matrix can be optimally determined based on
the instantaneous effective channel in a short-term time-scale.
Such a joint beamforming problem can be formulated into the
following optimization:
\begin{align}
\max_{\bm \psi}\ \ &\mathbb{E}\bigg\{\max_{\bm f} R \bigg\} \notag\\
s.t.\quad & |\psi_m| = 1 \ \ \forall m \in \{1,\cdots,M\}\notag\\
& \|\bm f\|^2 \le P_{\text{max}} \label{opt_beamforming}
\end{align}
where $\bm \psi=\text{diag}(\bm\Psi)$ with $\psi_m$ as its $m$th
element and $P_{\text{max}}$ is the total transmit power budget.
The optimization problem in~\eqref{opt_beamforming} has two
levels. The inner one is the rate-maximization problem with
respect to $\bm f$. This is implemented in each channel
realization with the given phase shift matrix $\bm\Psi$. The outer
one is an expectation maximization problem over the IRS phase
shift coefficients in which the expectation of the achievable rate
is taken over all possible channel realizations.

When $\bm \Psi$ is given, the optimal precoding vector is the
maximum-ratio transmission (MRT) which is explicitly given as
\begin{align}
\bm f = \sqrt{P_{\text{max}}}\frac{(\bm h^H\bm\Psi\bm G)^H}{\|\bm
h^H\bm\Psi\bm G\|_2} \label{opt_w}
\end{align}
Substituting \eqref{opt_w} into~\eqref{opt_beamforming}, we arrive
at the following optimization problem which is only related to
$\bm\Psi$:
\begin{align}
\max_{\bm \psi}\ \ &\mathbb{E}\left\{\log_2(1 + P_{\text{max}}\|\bm h^H\bm\Psi\bm G\|^2/\sigma^2) \right\} \notag\\
s.t.\quad & |\psi_m| = 1 \ \ \forall m \in \{1,\cdots,M\}
\label{opt_psi}
\end{align}
Directly solving~\eqref{opt_psi} is intractable. The major reason
is that its cost function is the expectation of a logarithmic
function, which in general does not have an explicit expression.
To handle this issue, we resort to maximize its tight upper bound,
which is given by~\cite{ZhaoWu21}
\begin{align}
&\mathbb{E}\{\log_2(1 + P_{\text{max}}\|\bm h^H\bm\Psi\bm G\|^2/\sigma^2) \}\notag\\
&\qquad\qquad\le \log_2(1 + P_{\text{max}}\mathbb{E}\{\|\bm
h^H\bm\Psi\bm G\|^2\}/\sigma^2) \label{upp_bound}
\end{align}
Note that the upper bound shown in~\eqref{upp_bound} is
sufficiently tight and thus is a good approximation of the
original objective function, especially when
$P_{\text{max}}/\sigma^2$ is large. Such an optimization trick has
been used in existing literatures, e.g.\cite{ZhaoWu21,HuWang20}. Maximizing the
upper bound yields the following optimization:
\begin{align}
\max_{\bm \psi}\ \ &\mathbb{E}(\|\bm h^H\bm\Psi\bm G\|^2)\notag\\
s.t.\quad & |\psi_m| = 1 \ \ \forall m \in \{1,\cdots,M\}
\label{opt_upt1}
\end{align}
where $\mathbb{E}(\|\bm h^H\bm\Psi\bm G\|^2)$ is given by
\begin{align}
\mathbb{E}(\|\bm h^H\bm\Psi\bm G\|^2) &= \mathbb{E}(\|\bm \psi^T \bm H\|^2)\notag\\
&= \bm \psi^T\mathbb{E}(\bm H\bm H^H)\bm \psi^*\notag\\
& = \bm \psi^T \bm {\bar R}_h \bm \psi^* \label{opt_upt1_cost}
\end{align}
in which $\bm {\bar R}_h$ is the covariance matrix of the cascade
channel $\bm H$. Note that $\bm {\bar R}_h =\mathbb{E}(\bm H\bm
H^H) $ and $\bm R_h = \mathbb{E}(\text{vec}(\bm H)\text{vec}(\bm
H)^H)$. Therefore, $\bm {\bar R}_h$ can be directly obtained from
$\bm R_h$. Specifically, the $(i,j)$th element of $\bm{\bar R}_h$
can be calculated by
\begin{align}
\bm{\bar R}_h(i,j) = \sum_{k=1}^N \bm R_h(i+(k-1)M,j+(k-1)M)
\end{align}
for $1\le i,j\le M$.

The optimization problem~\eqref{opt_upt1} with its cost function
given in~\eqref{opt_upt1_cost} is a nonconvex quadratically
constrained quadratic problem, and it can be relaxed as the
following semidefinite programming (SDP) problem
\begin{align}
\max_{\bm V} \ &\text{Tr}(\bm {\bar R}_h \bm { V})\notag\\
s.t.\  &  { V}_{ii} = 1\notag \\
\qquad & \bm { V}\succcurlyeq 0\notag\\
\qquad & \text{rank}(\bm{ V}) = 1 \label{opt_r1}
\end{align}
where $\bm { V} = \bm \psi^*\bm \psi^T $ is a rank-one matrix and ${
V}_{ii}$ is the $i$th diagonal element of $\bm  V$. When ignoring
the rank-one constraint, \eqref{opt_r1} can be immediately solved
by the off-the-shelf SDR programming toolboxes, e.g., the
CVX. After obtaining the optimal solution $\bm { V}$, we
need to find a feasible solution $\bm {v}$ from $\bm { V}$. One
efficient approach is the Gaussian randomization approximation
solution~\cite{SoZhang07}, and the detailed procedures can be found within.


\section{Simulation Results}
\label{s7}
In this section, we present simulation results to illustrate the
effectiveness of the proposed low-rank PSD Toeplitz-structured CCM
(LRT-CCM) estimation method. We compare our method with the
conventional CCM estimation approach which estimates the channel
at each time frame and then reconstructs the CCM with these
estimated channel samples. Such an approach is referred to as the
conventional CCM estimation method. For this approach, we use the
compressed sensing-based method~\cite{WangFang20-3} to estimate the channel at
each time frame.



In our simulations, the BS employs a uniform linear array (ULA) of
$N = 8$ antennas and the IRS is a planar array with $M_v\times M_h
=16\times 16$ reflecting elements. The three-dimensional
coordinates of the BS, the IRS, and the user are set to
$(5,0,10)$, $(0,50,20)$ and $(10,60,1.8)$, respectively. The
BS-IRS channel and the IRS-user channel are generated according
to~\eqref{bs_irs} and~\eqref{irs_ue}, respectively. The number of
the signal paths is set to $L=P=3$, and each channel comprises an
LOS path and two NLOS paths. The corresponding angles (including
AoAs and AoDs) of the LOS paths are determined by the geometry
configuration, and the associated complex gains of the LOS paths
are generated according to a complex Gaussian distribution
$\mathbb{CN}(0,10^{-0.1\kappa})$, where $\kappa = 61.4 +
29.2\log_{10}(d)+\epsilon$ with $d$ denoting the length of the
path and $\epsilon$ being a random variable following
$\mathbb{N}(0,(8.7 \text{dB})^2)$. The angles associated with the
NLOS paths are randomly selected from the interval $[-\pi,\pi]$.
The channel coefficients of these NLOS paths follow a distribution
$\mathbb{CN}(0,\delta^2)$, where $\delta$ is determined by the
Rician factor (i.e., the ratio of the energy of the LOS path to
that of all NLOS paths). We set the Rician factor to $10$ dB and
the transmitted power to $P_{\max} = 30\text{dBm}$. The
signal-to-noise ratio (SNR) is defined as
\begin{align}
\text{SNR} =
\mathbb{E}\left(10\log_{10}\left({\varsigma^2}/{\sigma^2}\right)\right)
\end{align}
where $\varsigma^2$ is the received signal power.

We evaluate the CCM estimation performance via the relative
efficiency metric (REM), which is widely adopted~\cite{HaghighatshoarCaire16,ParkPark17,AnjinappaGurbuz20} and
defined as
\begin{align}
\eta = \frac{\text{tr}(\bm U_1^H\bm R_h \bm U_1)}{\text{tr}(\bm
U_2^H\bm R_h \bm U_2)}
\end{align}
where $\bm U_1$ and $\bm U_2$ are, respectively, the matrices
constructed by the eigenvectors of the estimated CCM $\bm{\hat
R}_h$ and the eigenvectors of the true CCM $\bm R_h$. Clearly, a
higher value of $\eta$ indicates a more accurate CCM estimate. The
value $1-\eta$ means the lost of the signal power due to the
mismatch between the optimal beamformer and the estimated one. All
results are averaged over 50 independent Monte Carlo runs.



Fig.~\ref{f1}(a) plots the REMs of different methods as a function
of the SNR, where the number of time slots $J$ is set to $120$ and
the number of time frames $T$ is set to $100$. We see that our
proposed method presents a substantial performance improvement
over the conventional CCM estimation method, particularly in the
low SNR regime. In fact, our proposed method can still provide a
reliable CCM estimate even when the SNR is below -10dB, whereas
the conventional CCM estimation method performs poorly in such a
low SNR region. In Fig.~\ref{f1}(b), we plot the achievable rate
attained by the two-timescale beamforming scheme based on the
estimated CCM. To better evaluate the performance, we also include
the achievable rates attained by the two-timescale beamforming
scheme based on the true CCM, the two-timescale beamforming scheme
in which the reflecting coefficients are randomly chosen from a
unit circle (referred to as random passive beamforming), and the
joint beamforming scheme~\cite{YuXu19} that utilizes the true
I-CSI. Note that the beamforming approach~\cite{YuXu19} that
exploits the I-CSI provides an upper bound on the performance that
is achievable by any two-timescale beamforming schemes.

Several points can be made from Fig.~\ref{f1}(b). First, our
proposed method incurs only a very mild performance loss as
compared with the two-timescale beamforming scheme based on the
true CCM. This result indicates that the proposed method can yield
a CCM estimate that is good enough for subsequent beamforming.
Second, the two-timescale beamforming scheme can achieve
performance close to that of the beamforming method that utilizes
the I-CSI, which demonstrates the effectiveness of the
two-timescale beamforming scheme. Lastly, all methods present a
substantial performance advantage over the random passive
beamforming scheme.

\begin{figure*}[!t]
\centering \subfloat[REM]{
\includegraphics[width=0.48\textwidth]{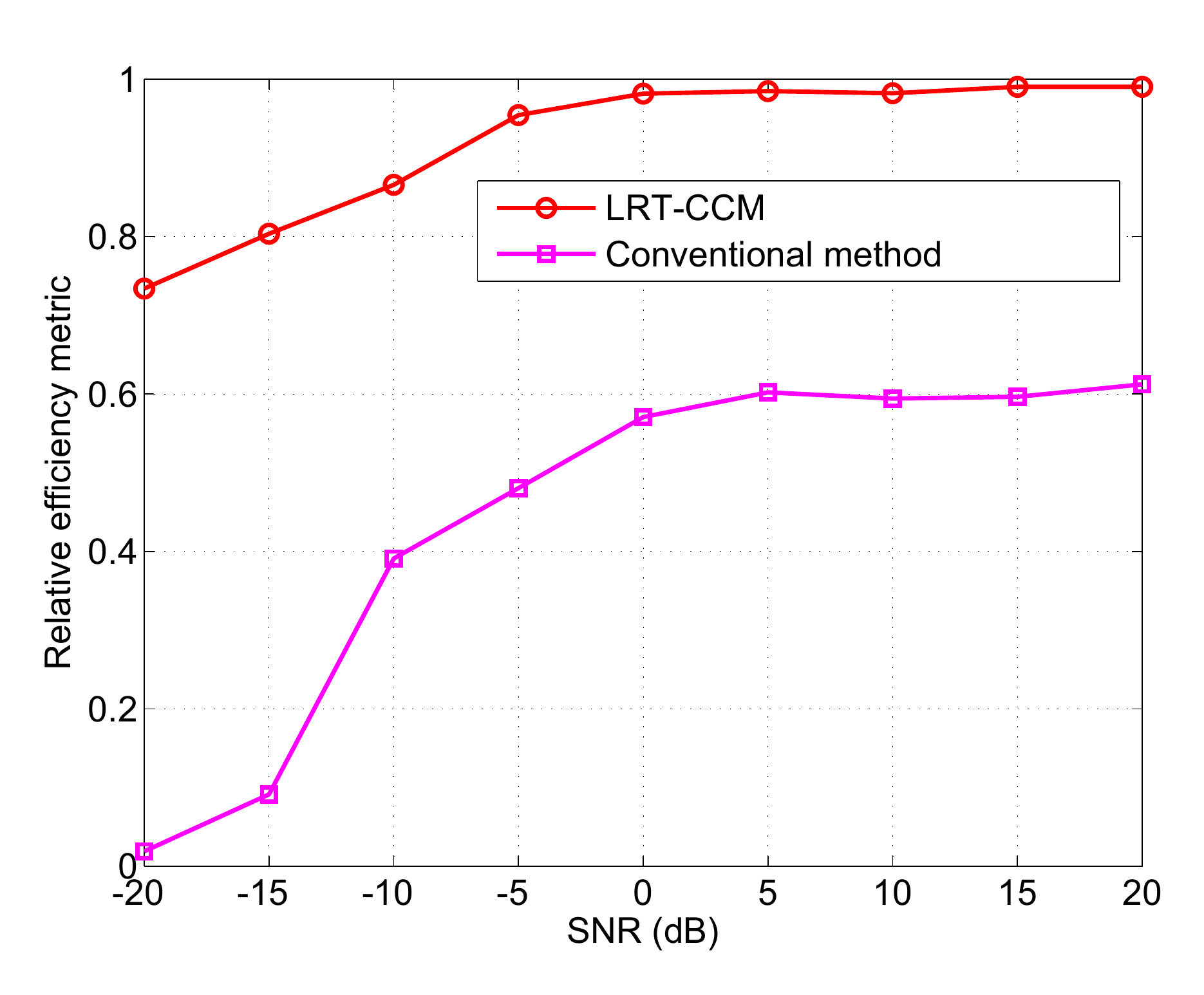}}\quad
\subfloat[Achievable
rate]{\includegraphics[width=0.48\textwidth]{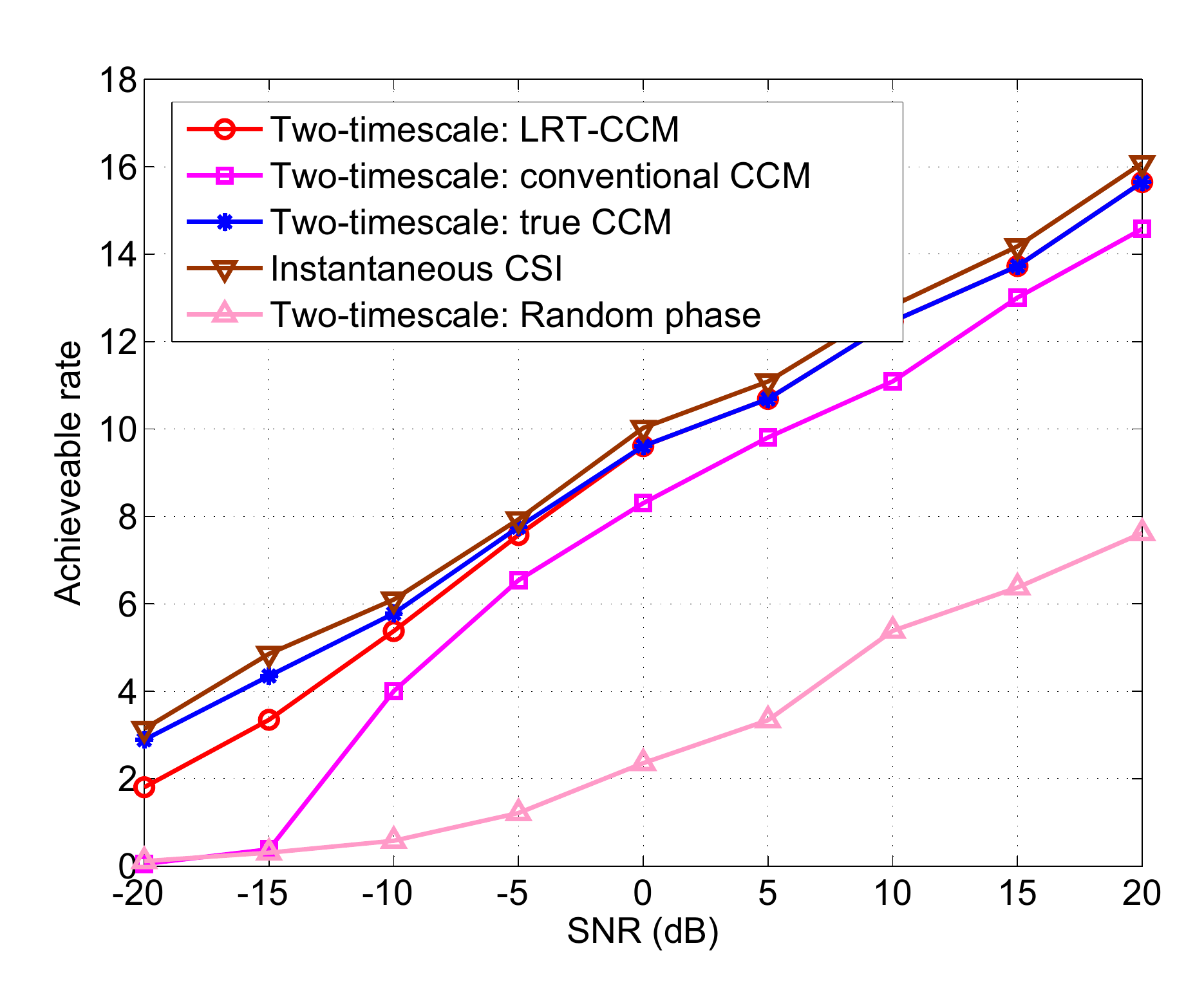}\label{f1b}}
\caption{REMs and achievable rates of respective methods, where we
set $T=100$ and $J = 120$.} \label{f1}
\end{figure*}

\begin{figure*}[!t]
\centering
\subfloat[REM]{\includegraphics[width=0.48\textwidth]{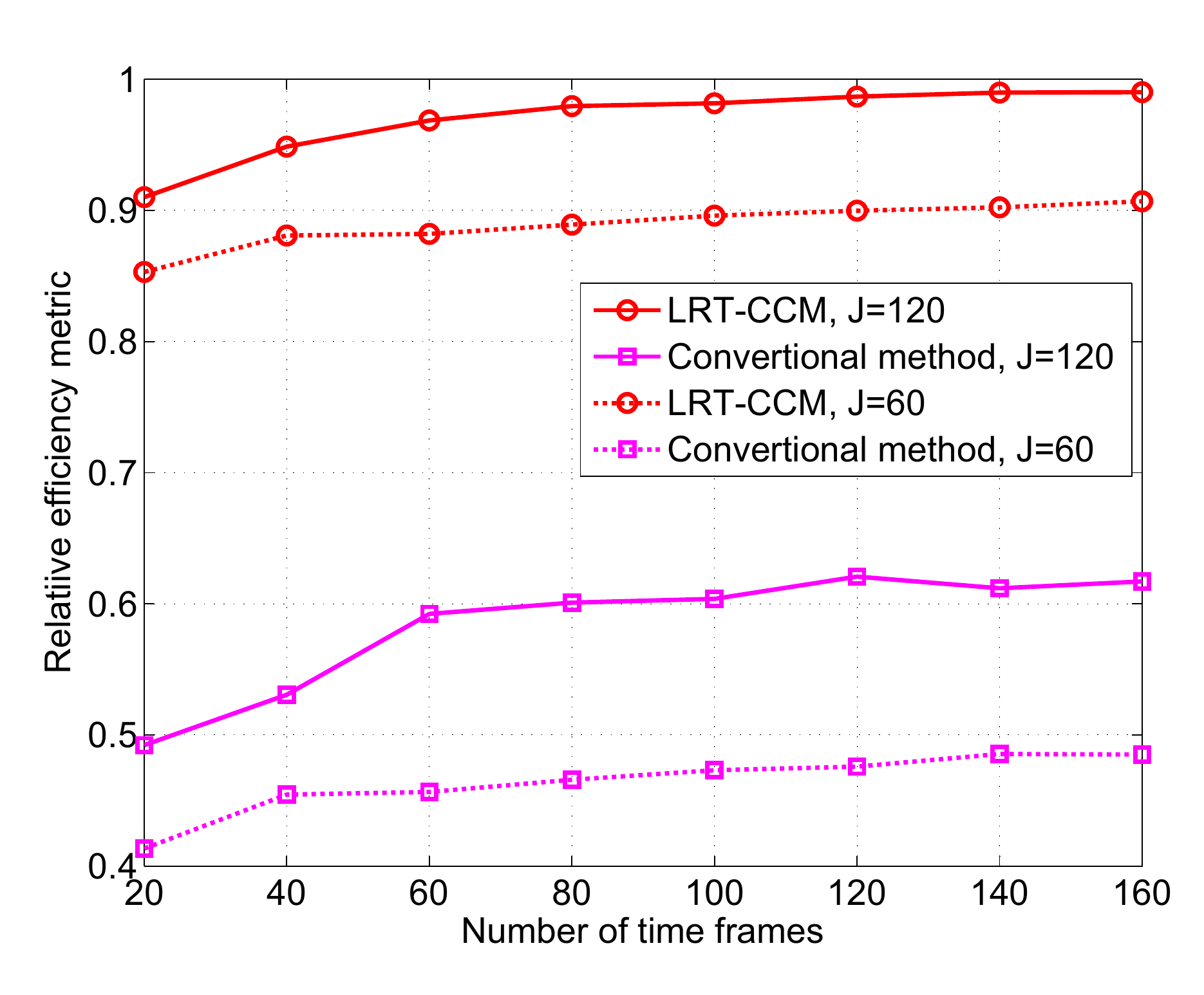}}\quad
\subfloat[Achievable
rate]{\includegraphics[width=0.48\textwidth]{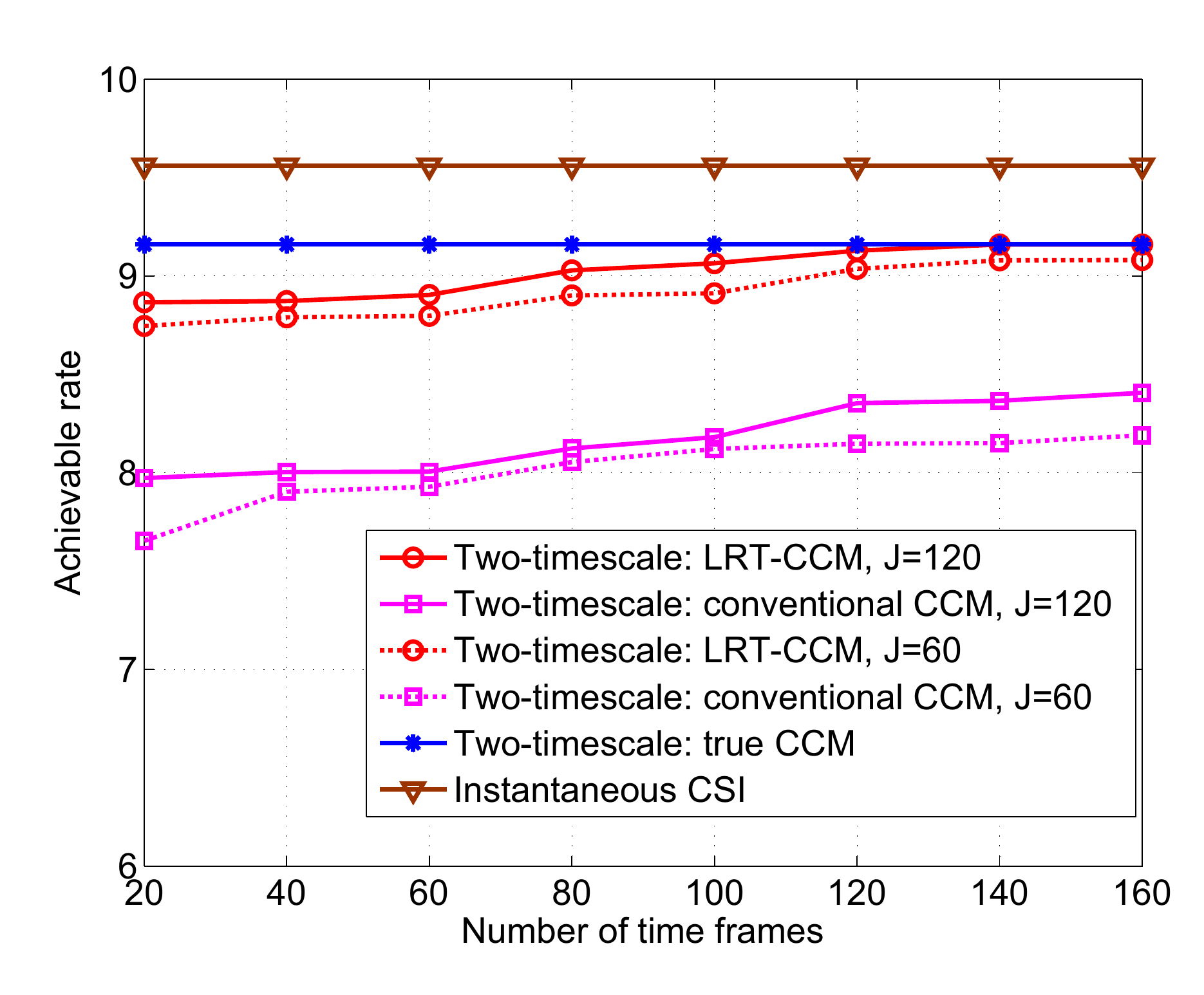}}
\caption{REMs and achievable rates of respective methods, where we
set $\text{SNR}=0$dB, $J = 60$ and $J = 120$.} \label{f2}
\end{figure*}

Next, we examine the impact of the number of time frames on the
estimation and beamforming performance. Fig.~\ref{f2} plots the
performance of respective methods as a function of $T$, where we
set SNR to $0$dB, and $J$ is set to $J=60$ and $J=120$,
respectively. It can be observed from Fig.~\ref{f2} that a small
value of $T$, say $T=20$ is sufficient to achieve a decent
performance for our proposed method. Increasing the number of time
frames can lead to better performance for both methods, but the
performance improvement is very limited. Since the total number of
measurements required for training is $TJ$, this result suggests
that our proposed method can provide a reliable CCM estimate using
a training overhead as small as $TJ=20\times 60=1200$.
Fig.~\ref{f3} illustrates the effect of the number of time slots
on the estimation and beamforming performance of respective
methods, where we set $\text{SNR}=0\text{dB}$, and $T$ is set to
$T=40$ and $T=100$, respectively. We see that increasing the
number of time slots leads to better performance. Nevertheless, a
small value of $J$, say, $J=60$, is enough to provide a decent
performance for our proposed method. Again, this result
demonstrates the ability of the proposed method in providing an
accurate CCM using a small amount of training overhead.



\begin{figure*}[!t]
\centering
\subfloat[REM]{\includegraphics[width=0.48\textwidth]{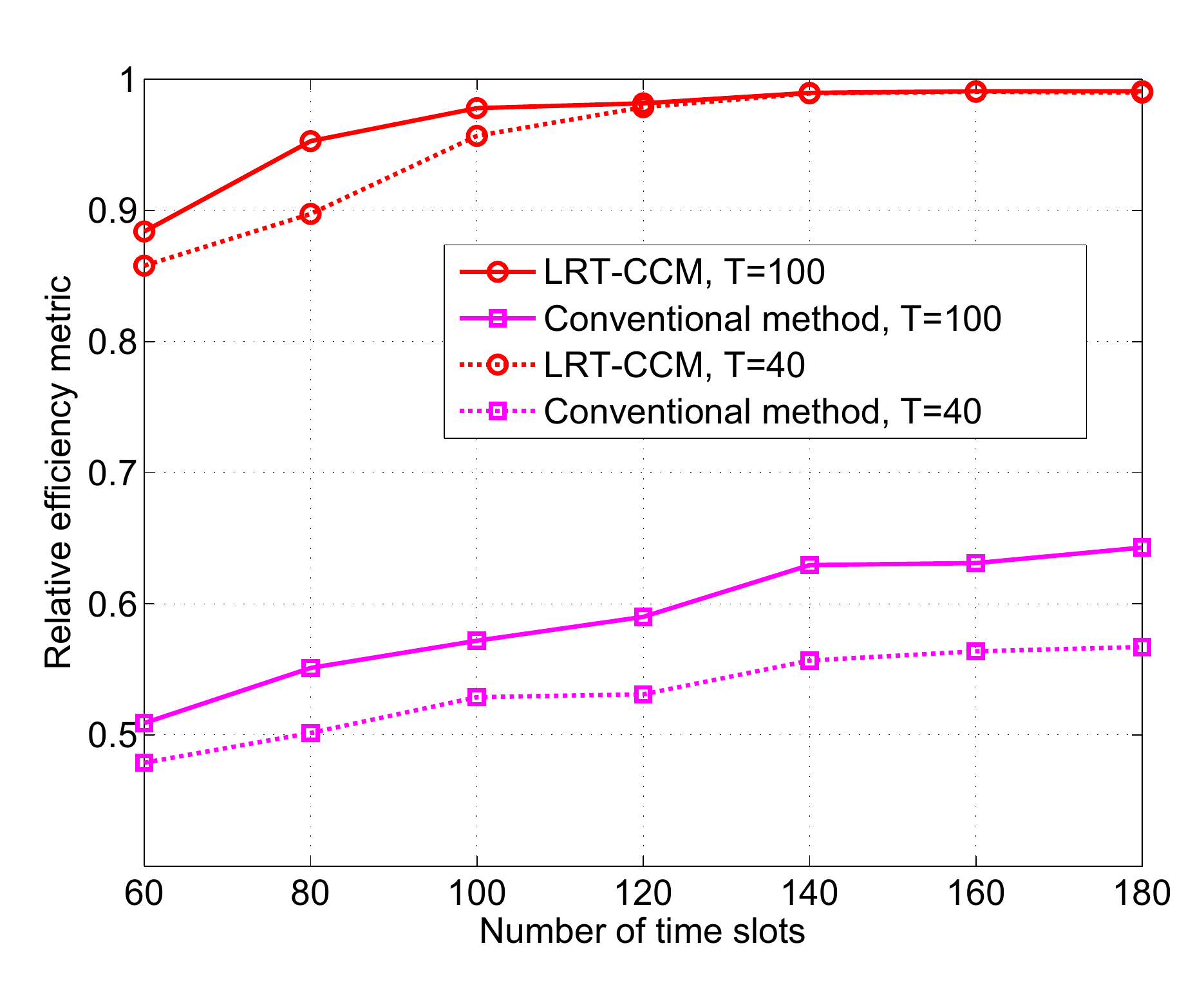}}\quad
\subfloat[Achievable
rate]{\includegraphics[width=0.48\textwidth]{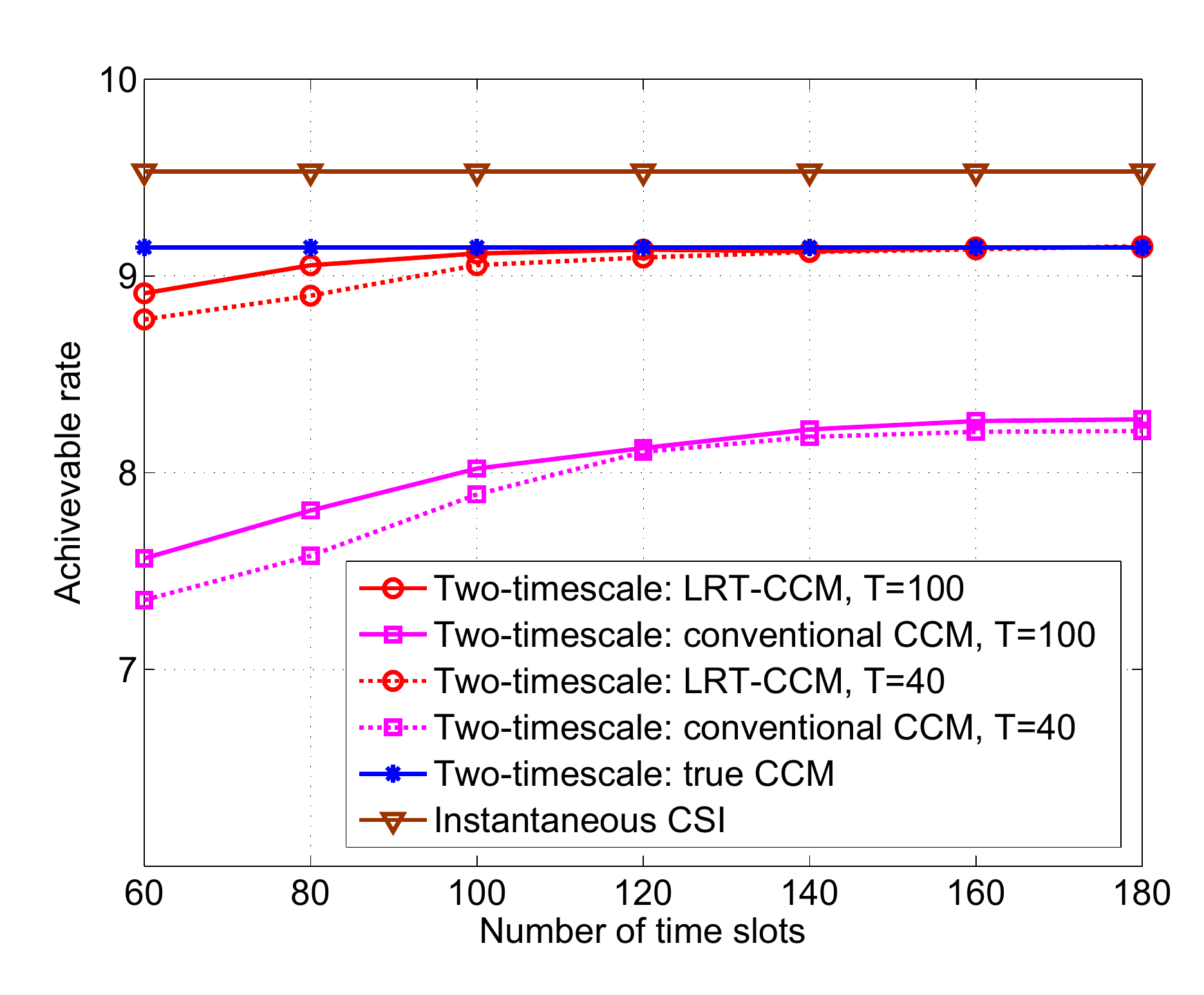}}
\caption{REMs and achievable rates of respective methods, where we
set $\text{SNR}=0$dB, $T=40$ and $T = 100$.} \label{f3}
\end{figure*}

\section{Conclusions}
\label{s8} In this paper, we considered the CCM estimation for
IRS-assisted mmWave communication systems. We exploited the
low-rank property and PSD 3-level Toeplitz structure of the CCM
and formulated the CCM estimation problem as a convex SDP problem,
which was further efficiently solved by an ADMM algorithm. In
addition, we analyzed the estimation performance of the proposed
solution, as well as the training overhead required to obtain a
reliable estimate of the CCM. Lastly, we discussed how to perform
the two-timescale beamforming based on the estimated CCM.
Simulation results showed that our proposed method can provide a
reliable CCM estimate using a small amount of training overhead.

\appendices

\section{Proof of Theorem 1}
\label{A1} Based on the definition, the trace of a PSD matrix is equivalent to its nuclear norm. For
simplicity, we consider the following equivalent form
of~\eqref{sdp}:
\begin{align}
\bm{\hat{ R}}_h = \arg\min_{\bm{{ R}}_h} \frac{1}{2} \left\|\bm
{\hat R}_y - \bm W\bm{{ R}}_h\bm W^H\right\|_F^2 +\lambda \|\bm{{
R}}_h\|_* \label{opt1}
\end{align}
where $\|\cdot\|_*$ denotes the nuclear norm.

In order to prove Theorem~\ref{t1}, we first introduce the
following theorem~\cite{NegahbanRavikumar12}.
\begin{theorem}
For the convex optimization problem
\begin{align}
\widehat{\bm \Theta}_{\lambda_{n}} \in \arg \min _{\bm
\Theta}\left\{\mathcal{L}\left(\bm \Theta \right)+\lambda_{n}
\mathcal{R}(\bm \Theta)\right\} \label{pro_theorem2}
\end{align}
where $\lambda_{n}> 0$ is a user-defined regularization parameter
and $\mathcal{R}(\cdot)$ is a norm. Suppose that $\mathcal{L}$ is
a convex and differentiable function, and consider any optimal
solution $\widehat{\bm \Theta}$ to the aforementioned optimization
problem with a strictly positive regularization parameter
satisfying
\begin{align}
\lambda_n \ge 2\mathcal{R}^*(\nabla\mathcal{L}\left(\bm \Theta^*
\right))
\end{align}
where $\mathcal{R}^*(\cdot)$ is the dual norm of
$\mathcal{R}(\cdot)$ and $\bm \Theta^*$ is the unknown true value.
Denote $\mathcal{M}\subseteq \bar{\mathcal {M}}$ as the subspace
to capture the constraints specified by the norm-based regularizer
and $\bar{\mathcal {M}}^{\perp}$ as the orthogonal complement of
space $\bar{\mathcal {M}}$. Then for any pair
$\left(\mathcal{M},\bar{\mathcal {M}}^{\perp}\right)$ over which
$\mathcal{R}$ is decomposable{\footnote{A norm-based regularizer
$\mathcal{R}$ is decomposable with respect to
$\left(\mathcal{M},\bar{\mathcal {M}}^{\perp}\right)$ if
\begin{align*}
\mathcal{R}(\bm\Theta + \bm\Gamma) = \mathcal{R}(\bm\Theta)+\mathcal{R}(\bm\Gamma)
\end{align*}
for all $\bm\Theta\in \mathcal{M}$ and $\bm\Gamma\in\bar{\mathcal
{M}}^{\perp}$. Details can be found
in~\cite{NegahbanRavikumar12}.}}, the error $\bm \Delta =
\widehat{\bm \Theta}_{\lambda_{n}} - \bm \Theta^*$ belongs to the
set
\begin{align}
&\mathbb{C}\left(\mathcal{M}, \bar {\mathcal{M}}^{\perp} ; \bm\Theta^{*}\right) \notag\\
&\quad \triangleq \left\{\bm \Delta |
\mathcal{R}\left(\bm \Delta_{\bar{\mathcal{M}}^{\perp}}\right)\right.\left.\leq
3 \mathcal{R}\left(\bm \Delta_{\bar{\mathcal{M}}}\right)+4
\mathcal{R}\left(\bm
\Theta_{\mathcal{M}^{\perp}}^{*}\right)\right\} \label{1}
\end{align}
where $\mathcal{M}^{\perp}$ is the orthogonal complement of the
space $\mathcal{M}$. In~\eqref{1} $\bm\Delta_{\bar{\mathcal{M}}}$
denotes the projection of $\bm\Delta$ onto the subspace
$\bar{\mathcal{M}}$, which is defined as
\begin{align}
\bm\Delta_{\bar{\mathcal{M}}} =\arg\min_{\bm\Delta_1\in \bar{\mathcal{M}}}\|\bm \Delta - \bm \Delta_1\|_F^2
\end{align}
${\bm \Delta}_{\bar{\mathcal{M}}^{\perp}}$ and $\bm
\Theta_{\mathcal{M}^{\perp}}^{*}$ can be similarly defined (and
hence we omit them here). Specifically, $\mathcal{R}\left(\bm
\Theta_{\mathcal{M}^{\perp}}^{*}\right) = 0$ when $\bm\Theta\in
\mathcal{M}$, and under such a circumstance~\eqref{1} turns to be
\begin{align}
\mathbb{C}\left(\mathcal{M}, \bar{\mathcal{M}}^{\perp} ; \bm
\Theta^{*}\right) \triangleq\left\{\bm \Delta |
\mathcal{R}\left(\bm \Delta_{\bar{\mathcal{M}}^{\perp}}\right)\right.\left.\leq
3 \mathcal{R}\left(\bm \Delta_{\bar{\mathcal{M}}}\right)\right\}
\label{eqn1}
\end{align}
\end{theorem}
\begin{proof}
See~\cite{NegahbanRavikumar12}.
\end{proof}


Theorem 2 implies that with a proper regularization parameter and
$\bm\Theta \in \mathcal{M}$, the estimation error of the
problem~\eqref{pro_theorem2} satisfies (\ref{eqn1}). We know that
$\text{rank}(\bm R_h) = r$ and the SVD of $\bm R_h$ is given as
$\bm R_h = \bm U\bm\Sigma\bm V^H$. Let $\text{row}(\bm R_h)$ and
$\text{col}(\bm R_h)$ denote the row and column space of $\bm R_h$
respectively, and meanwhile let $\bm U^r$ and $\bm V^r$ be the
first $r$ columns of $\bm U$ and $\bm V$. We now can define the
subspace $\mathcal{M}$ and $\bar{\mathcal{M}}^{\perp}$ as
\begin{align}
&\mathcal{M} = \bar{\mathcal{M}}\triangleq \{\bm R|\text{row}(\bm R)=\bm V^r,\text{col}(\bm R)=\bm U^r\}\\
&\bar{\mathcal{M}}^{\perp}\triangleq\{\bm R|\text{row}(\bm
R)={(\bm V^r)}^{\perp},\text{col}(\bm R)={(\bm U^r)}^{\perp}\}
\end{align}
Utilizing these defined subspaces, we can conclude that $\bm R_h
\in \mathcal{M}$, and meanwhile the estimation error of $\bm R_h$,
i.e., $\bm \Delta\triangleq \bm{\hat R}_h - \bm R_h$, can be
decomposed into two parts, i.e.,
\begin{align}
\bm \Delta = \bm \Delta_1 + \bm \Delta_2
\end{align}
with $\text{rank}(\bm \Delta_1) = r$, $\bm \Delta_1\in\mathcal{M}$
and $\bm \Delta_2\in\bar{\mathcal{M}}^{\perp}$.

In our problem, $\mathcal{L}\left(\bm \Theta \right) =
\frac{1}{2}\|\bm {\hat{R}}_{y} - \bm W\bm R_h\bm W^H\|_F^2$.
Therefore, we have
\begin{align}
\nabla\mathcal{L}\left(\bm R_h \right) = \frac{1}{2}\left[\bm
W^H(\bm W\bm R_h\bm W^H-\bm {\hat {R}}_{y})\bm W\right]^T
\end{align}
In addition, it is clear that the dual norm of the nuclear norm is
the spectral norm. Therefore, if the regularization parameter
$\lambda$ satisfies
\begin{align}
\lambda \ge 2\mathcal{R}^*(\nabla\mathcal{L}\left(\bm R_h \right)) &= \|\bm W^H(\bm {\hat {R}}_{y}-
\bm W \bm R_h\bm W^H)\bm W\|_2\notag\\
&= \|\bm W^H(\bm {\hat {R}}_{y}- \bm R_{y})\bm W\|_2
\label{lambda_setup}
\end{align}
then the estimation error, based on Theorem 2, belongs to the set
\begin{align}
\bm \Delta &\triangleq \bm{\hat R}_h - \bm R_h =\mathbb{T}_3(\hat{\bm V}) -
\mathbb{T}_3({\bm V}) =  \mathbb{T}_3(\hat{\bm V}-{\bm V})\notag\\
&\in \{\bm \Delta|\bm\Delta = \bm \Delta_1 + \bm \Delta_2,\ \|\bm
\Delta_2\|_*\le 3\|\bm\Delta_1\|_*\} \label{three_times}
\end{align}

From~\eqref{lambda_setup}, we can see that the choice of $\lambda$
depends on the value of $\|\bm W^H(\bm {\hat {R}}_{y}-\bm R_y)\bm
W\|_2$. The following lemma provides an upper bound on $\|\bm
W^H(\bm {\hat {R}}_{y}-\bm R_y)\bm W\|_2$.
\begin{lemma} \label{lemma1}
Given $T$ observation samples $\{\bm y_t\}_{t=1}^T$, $\bm{\hat{
R}}_y$ is obtained via $\frac{1}{T} \sum_{t=1}^T\bm y_t\bm y_t^H$.
Then with probability at least $1-4T^{-1}$, we have
\begin{align}
&\|\bm W^H (\bm {\hat{R}}_{y} - \bm R_y)\bm W\|_2\notag\\
&\qquad\qquad\le \delta\triangleq c \|\bm W\|_F^2 \|\bm R_y
\|_2\text{max}\{\sqrt{ \tilde{\delta}},\tilde{\delta}\}
\end{align}
where $\tilde{\delta}$ is given in~\eqref{sigma1}.
\end{lemma}
\begin{proof}
See Appendix~\ref{ap2}.
\end{proof}

Since $\bm{\hat R}_h$ is the optimal solution to (\ref{opt1}), we
have
\begin{align}
&\frac{1}{2}\left\|\bm {\hat {R}}_{y} - \bm W \bm{\hat R}_h \bm W^H\right\|_F^2 +
\lambda \|\bm{\hat R}_h\|_* \notag\\
&\qquad \le \frac{1}{2}\left\|\bm {\hat {R}}_{y} - \bm W {\bm R}_h
\bm W^H\right\|_F^2 + \lambda \|{\bm R}_h\|_*
\end{align}
which can be converted to
\begin{align}
\left\|\bm {\hat {R}}_{y} - \bm W \bm{\hat R}_h \bm W^H\right\|_F^2 &-
\left\|\bm {\hat {R}}_{y} - \bm W {\bm R}_h \bm W^H\right\|_F^2\notag\\
 &\qquad\le 2\lambda\Big(\|{\bm R}_h\|_* - \|\bm{\hat R}_h\|_*\Big)
\label{2}
\end{align}
Recalling $\bm \Delta= \bm{\hat R}_h - \bm R_h$, the left side
of~\eqref{2} can be further expressed as
\begin{align}
&\left\|\bm {\hat {R}}_{y} - \bm W \hat{\bm R}_h \bm W^H\right\|_F^2 -
\left\|\bm {\hat {R}}_{y} - \bm W {\bm R}_h \bm W^H\right\|_F^2\notag\\
& = \left\|\bm {\hat {R}}_{y} - \bm W ({\bm R}_h + \bm \Delta ) \bm W^H\right\|_F^2 -
\left\|\bm {\hat {R}}_{y} - \bm W {\bm R}_h \bm W^H\right\|_F^2\notag\\
& = 2\langle-\bm {\hat {R}}_{y} + \bm W {\bm R}_h\bm W^H, \bm W\Delta\bm W^H \rangle +
\left\|\bm W\bm \Delta\bm W^H\right\|_F^2 \notag \\
&= 2\langle \bm W^H(-\bm {\hat {R}}_{y} + \bm W {\bm R}_h\bm
W^H)\bm W,\bm \Delta \rangle + \left\|\bm W\bm \Delta\bm
W^H\right\|_F^2 \label{3}
\end{align}
Substituting~\eqref{3} into~\eqref{2}, we have
\begin{align}
&\left\|\bm W\bm \Delta\bm W^H\right\|_F^2 \notag\\
&\le 2\langle \bm W^H(\bm {\hat {R}}_{y} - \bm W {\bm R}_h\bm W^H)\bm W,\bm \Delta \rangle +
2\lambda\Big(\|{\bm R}_h\|_* - \|\hat{\bm R}_h\|_*\Big)\notag\\
&\overset{(a)}{\le} 2\|\bm W^H(\bm{\hat {R}}_{y} - \bm W {\bm R}_h\bm W^H)\bm W \| \|\bm \Delta \|_*\notag\\
&\qquad\qquad\qquad\qquad\qquad+ 2\lambda\Big(\|\bm{\hat R}_h+\bm\Delta\|_* - \|\bm{\hat R}_h\|_*\Big)\notag\\
&\overset{(b)}{\le}2 \|\bm W^H(\bm{\hat {R}}_{y} - \bm W {\bm R}_h\bm W^H)\bm W \| \|\bm \Delta\|_*+
2\lambda \|\bm \Delta\|_*\notag\\
& \overset{(c)}{\le} 4\lambda \left\|\bm \Delta\right\|_*
\label{use1}
\end{align}
where $(a)$ follows from Holder's inequality, $(b)$ comes from the
triangle inequality, and $(c)$ follows from the assumption
$\lambda \ge \delta$. Furthermore, we have
\begin{align}
\left\|\bm \Delta\right\|_* &= \left\|\bm \Delta_1 + \bm \Delta_2\right\|_* \le
\left\|\bm \Delta_1\right\|_*+\left\| \bm \Delta_2\right\|_*\notag\\
&\overset{(a)}{\le} \left\|\bm \Delta_1\right\|_*+3\left\| \bm \Delta_1\right\|_*\notag\\
& \overset{(b)}{\le} 4 \sqrt{r}\|\bm \Delta_1\|_F \notag\\
&\le 4 \sqrt{r} \left\|\bm \Delta\right\|_F = 4\sqrt{r}
\left\|\mathbb{T}_3(\hat{\bm V} - \bm V)\right\|_F\notag\\
&\le 4 \sqrt{rNM} \|\hat{\bm V} - \bm V\|_F \label{use2}
\end{align}
where $(a)$ follows from~\eqref{three_times}, and $(b)$ is a
result of the relationship between the nuclear norm and F-norm of
a rank-$r$ matrix. Putting~\eqref{use1} and~\eqref{use2} together
results in
\begin{align}
\left\|\bm W\bm \Delta\bm W^H\right\|_F^2 \le 16\lambda\sqrt{rNM}
\|\hat{\bm V} - \bm V\|_F \label{re1}
\end{align}

Furthermore, we utilize the following lemma to find a lower bound
of $\left\|\bm W\bm \Delta\bm W^H\right\|_F^2$.

\begin{lemma}
\label{lemma2} Consider a 3-level Toeplitz matrix $\mathbb{T}_3(
\bm{{ X}})$ with $\bm X\in\mathbb{C}^{(2I_1-1)\times
(2I_2-1)\times (2I_3-1)}$ and the matrix $\bm
W\in\mathbb{C}^{I^2\times I_1 I_2 I_3}$. If
\begin{align}
I\ge\sqrt{(2I_1-1) (2I_2-1) (2I_3-1)} \label{con1}
\end{align}
then with high probability there exists a full-column rank
transforming matrix $\bm {\check W}$ of $\boldsymbol{W}$ such that
\begin{align}
\|\bm W\mathbb{T}_3( \bm{{X}})\bm W^H\|_F^2\ge
\sigma_{\text{min}}^2(\bm {\check W})\|\bm{{X}}\|_F^2
\end{align}
where $\sigma_{\text{min}}(\bm {\check W})$ is the smallest
singular value of $\bm {\check W}$.
\end{lemma}
\begin{proof}
See Appendix~\ref{ap3}.
\end{proof}

Since we have
\begin{align}
J\ge u\triangleq\sqrt{(2N-1)(2M_v-1)(2M_h-1)}
\end{align}
and meanwhile $\bm \Delta = \mathbb{T}_3(\bm{\hat V}-\bm V)$ is a
3-level Toeplitz matrix, based on Lemma~\ref{lemma2}, we know that, with high probability the following holds
\begin{align}
\|\bm W\bm \Delta\bm W^H\|_F^2 \ge  \sigma_{\text{min}}^2(\bm
{\check W} )\|\hat{\bm V} - \bm V\|_F^2 \label{re2}
\end{align}
where $\bm {\check W} $ is the transforming matrix of $\bm W$.
Substituting~\eqref{re2} into~\eqref{re1} leads to
\begin{align}
\|\bm{\hat V} - \bm
V\|_F\le\frac{16\lambda\sqrt{rNM}}{\sigma_{\text{min}}^2(\bm
{\check W})}
\end{align}
As a result, the average per-entry RMSE is given by
\begin{align}
\frac{1}{u}\|\bm{\hat V} - \bm V\|_F\le
\frac{16\lambda\sqrt{r}}{\sigma^2_{\text{min}}(\bm {\check W})}
\frac{\sqrt{NM}}{u}
\end{align}
which completes the proof.

\section{Proof of Lemma \ref{lemma1}}
\label{ap2}
The received signal $\bm y_t$ follows a complex Gaussian
$\mathbb{CN}(0,\bm R_y)$. Therefore, according to Theorem 2.2
in~\cite{BuneaXiao15} we know that, with probability at least
$1-4T^{-1}$, the following inequality holds
\begin{align}
\| \bm {\hat {R}}_y- \bm R_y\|_2\le c\|\bm
R_y\|_2\text{max}\{\sqrt{\tilde{\delta}},\tilde{\delta}\}
\end{align}
where $c$ is a constant. Furthermore, we have
\begin{align}
&\|\bm W^H ( \bm {\hat {R}}_y- \bm R_y)\bm W\|_2\notag\\
&\overset{(a)}{\le} \|\bm W^H ( \bm {\hat {R}}_y- \bm R_y)\bm W\|_F\notag\\
&\overset{(b)}{\le} \|\bm W\|_F^2 \|( \bm {\hat {R}}_y- \bm R_y)\|_2\notag\\
&\le c \|\bm W\|_F^2 \|\bm R_y
\|_2\text{max}\{\sqrt{\tilde{\delta}},\tilde{\delta}\}
\end{align}
where $(a)$ follows from the fact that $\|\bm A\|_2\le\|\bm A\|_F$
and $(b)$ follows from $\|\bm A\bm B\|_F\le \|\bm A\|_2\|\bm
B\|_F$.

\section{Proof of Lemma \ref{lemma2}}
\label{ap3}
According to the Kronecker product property, we have
\begin{align}
\text{vec}(\bm W\mathbb{T}_3( \bm{{ X}})\bm W^H) & = (\bm
W^*\otimes \bm W)\text{vec}(\mathbb{T}_3( \bm{X})) \label{5}
\end{align}
Since $\mathbb{T}_3( \bm{X})$ is a 3-level Toeplitz matrix, we can
express \eqref{5} as
\begin{align}
(\bm W^*\otimes \bm W)\text{vec}(\mathbb{T}_3( \bm{{ X}})) = \bm
{\check W} \bm x
\end{align}
where $\bm x=\text{vec}(\bm X)$ and $\bm {\check W}\in
\mathbb{C}^{I^2\times (2I_1-1) (2I_2-1) (2I_3-1)}$ is
constructed by combining those columns in $\bm W^*\otimes \bm W$
which correspond to the same element in $\mathbb{T}_3( \bm{X})$.
Furthermore, $\bm {\check W}$ is a full-column rank matrix with
high probability due to the fact that $\bm W$ is a random matrix
and the condition \eqref{con1} holds. Since $\bm {\check W}$ is a
full-column rank matrix, we have
\begin{align}
\|\bm {\check W} \bm x\|_2^2 \ge \sigma_{\text{min}}^2(\bm {\check
W})\|\bm x\|_2^2 \equiv \sigma_{\text{min}}^2 (\bm {\check
W})\|\bm{X }\|_F^2
\end{align}
Due to that fact that $\|\bm A\|_F^2 = \|\text{vec}(\bm A)\|_2^2$
for an arbitrary matrix $\bm A$, we have
\begin{align}
\|\bm W\mathbb{T}_3( \bm{X})\bm W^H\|_F^2 = \|\bm {\check W} \bm
x\|_2^2 \ge \sigma_{\text{min}}^2(\bm {\check W})\|\bm{X}\|_F^2
\end{align}
which completes the proof.

\section{Multi-level Toeplitz Matrix}
\label{ap4} Given a $d$-order tensor $ \bm{{
U}}_d\in\mathbb{C}^{(2k_d-1)\times\cdots\times (2k_1-1)}$, its
corresponding $d$-level Toeplitz matrix
$\mathbb{T}_d(\bm{{ U}}_d)$ is defined in the following
recursive manner. For $d=1$, $\bm{{ U}}_d$ is essentially
a vector $\bm u_1 \in\mathbb{C}^{2k_1-1}$, in which case
$\mathbb{T}_1(\bm u_1)$ is given by
\begin{align}
\mathbb{T}_1(\bm u_1) =
\begin{bmatrix}
\bm u_1(0)& \bm u_1(1)&\cdots&\bm u_1(k_1-1)\\
\bm u_1(-1)&\bm u_1(0)&\cdots&\bm u_1(k_1-2)\\
\vdots&\vdots&\ddots&\vdots\\
\bm u_1(1-k_1)&\bm u_1(2-k_1)&\cdots&\bm u_1(0)\\
\end{bmatrix}
\label{MT1}
\end{align}
For $d \ge 2$, Let $\bm{{ U}}_{d-1}(i)=\bm{{
U}}_d(i,:,\cdots,:)$ with $i$ from $-k_d+1$ to $k_d-1$, and denote
$\mathbb{T}_{d-1}(\bm{{U}}_{d-1}(i))$ as $\bm D(i)$. We
have
\begin{align}
\mathbb{T}_d(\bm{{U}}_{d}) =
\begin{bmatrix}
\bm D(0)& \bm D(1)&\cdots&\bm D(k_d-1)\\
\bm D(-1)&\bm D(0)&\cdots&\bm D(k_d-2)\\
\vdots&\vdots&\ddots&\vdots\\
\bm D(1-k_d)&\bm D(2-k_d)&\cdots&\bm D(0)\\
\end{bmatrix}
\label{MT2}
\end{align}



\begin{thebibliography}{10}
\providecommand{\url}[1]{#1}
\csname url@samestyle\endcsname
\providecommand{\newblock}{\relax}
\providecommand{\bibinfo}[2]{#2}
\providecommand{\BIBentrySTDinterwordspacing}{\spaceskip=0pt\relax}
\providecommand{\BIBentryALTinterwordstretchfactor}{4}
\providecommand{\BIBentryALTinterwordspacing}{\spaceskip=\fontdimen2\font plus
\BIBentryALTinterwordstretchfactor\fontdimen3\font minus
  \fontdimen4\font\relax}
\providecommand{\BIBforeignlanguage}[2]{{%
\expandafter\ifx\csname l@#1\endcsname\relax
\typeout{** WARNING: IEEEtran.bst: No hyphenation pattern has been}%
\typeout{** loaded for the language `#1'. Using the pattern for}%
\typeout{** the default language instead.}%
\else
\language=\csname l@#1\endcsname
\fi
#2}}
\providecommand{\BIBdecl}{\relax}
\BIBdecl

\bibitem{RanganRappaport14}
S.~Rangan, T.~S. Rappaport, and E.~Erkip, ``Millimeter-wave cellular wireless
  networks: potentials and challenges,'' \emph{Proceedings of the IEEE}, vol.
  102, no.~3, pp. 366--385, 2014.

\bibitem{WangFang20-1}
P.~Wang, J.~Fang, X.~Yuan, Z.~Chen, and H.~Li, ``Intelligent reflecting
  surface-assisted millimeter wave communications: Joint active and passive
  precoding design,'' \emph{IEEE Transactions on Vehicular Technology},
  vol.~69, no.~12, pp. 14\,960--14\,973, 2020.

\bibitem{WuZhang19}
Q.~Wu and R.~Zhang, ``Intelligent reflecting surface enhanced wireless network
  via joint active and passive beamforming,'' \emph{IEEE Transactions on
  Wireless Communications}, vol.~18, no.~11, pp. 5394--5409, 2019.

\bibitem{ZhengYou2022}
B.~Zheng, C.~You, W.~Mei, and R.~Zhang, ``A survey on channel estimation and
  practical passive beamforming design for intelligent reflecting surface aided
  wireless communications,'' \emph{IEEE Communications Surveys \& Tutorials},
  2022.

\bibitem{WangFang20-3}
P.~Wang, J.~Fang, H.~Duan, and H.~Li, ``Compressed channel estimation for
  intelligent reflecting surface-assisted millimeter wave systems,'' \emph{IEEE
  Signal Processing Letters}, vol.~27, pp. 905--909, 2020.

\bibitem{LiuGao20}
S.~Liu, Z.~Gao, J.~Zhang, M.~Di~Renzo, and M.-S. Alouini, ``Deep denoising
  neural network assisted compressive channel estimation for mmwave intelligent
  reflecting surfaces,'' \emph{IEEE Transactions on Vehicular Technology},
  vol.~69, no.~8, pp. 9223--9228, 2020.

\bibitem{WeiShen21}
X.~Wei, D.~Shen, and L.~Dai, ``Channel estimation for {RIS} assisted wireless
  communications--part {II}: An improved solution based on double-structured
  sparsity,'' \emph{IEEE Communications Letters}, vol.~25, no.~5, pp.
  1403--1407, 2021.

\bibitem{WangZhang21}
W.~Wang and W.~Zhang, ``Joint beam training and positioning for intelligent
  reflecting surfaces assisted millimeter wave communications,'' \emph{IEEE
  Transactions on Wireless Communications}, vol.~20, no.~10, pp. 6282--6297,
  2021.

\bibitem{WangFang22-1}
P.~Wang, J.~Fang, W.~Zhang, and H.~Li, ``Fast beam training and alignment for
  {IRS}-assisted millimeter wave/terahertz systems,'' \emph{IEEE Transactions
  on Wireless Communications}, vol.~21, no.~4, pp. 2710--2724, 2022.

\bibitem{WangFang22-2}
P.~Wang, J.~Fang, W.~Zhang, Z.~Chen, H.~Li, and W.~Zhang, ``Beam training and
  alignment for {RIS}-assisted millimeter wave systems: State of the art and
  beyond,'' \emph{IEEE Wireless Communications}, to appear.

\bibitem{LinJin21}
Y.~Lin, S.~Jin, M.~Matthaiou, and X.~You, ``Channel estimation and user
  localization for {IRS}-assisted {MIMO}-{OFDM} systems,'' \emph{IEEE
  Transactions on Wireless Communications}, 2021.

\bibitem{DeDe21}
G.~T. de~Ara{\'u}jo, A.~L. De~Almeida, and R.~Boyer, ``Channel estimation for
  intelligent reflecting surface assisted {MIMO} systems: A tensor modeling
  approach,'' \emph{IEEE Journal of Selected Topics in Signal Processing},
  vol.~15, no.~3, pp. 789--802, 2021.

\bibitem{ZhengWang22}
X.~Zheng, P.~Wang, J.~Fang, and H.~Li, ``Compressed channel estimation for
  {IRS}-assisted millimeter wave {OFDM} systems: A low-rank tensor
  decomposition-based approach,'' \emph{IEEE Wireless Communications Letters},
  pp. 1--1, 2022.

\bibitem{YangWang20}
F.~Yang, J.-B. Wang, H.~Zhang, C.~Chang, and J.~Cheng, ``Intelligent reflecting
  surface-assisted mmwave communication exploiting statistical {CSI},'' in
  \emph{2020 IEEE International Conference on Communications (ICC)}.\hskip 1em
  plus 0.5em minus 0.4em\relax Dublin, Ireland, 2020, pp. 1--6.

\bibitem{HuGao20}
X.~Hu, F.~Gao, C.~Zhong, X.~Chen, Y.~Zhang, and Z.~Zhang, ``An angle domain
  design framework for intelligent reflecting surface systems,'' in \emph{2020
  IEEE Global Communications Conference (GLOBECOM)}.\hskip 1em plus 0.5em minus
  0.4em\relax Taipei, Taiwan, 2020, pp. 1--6.

\bibitem{HuWang20}
X.~Hu, J.~Wang, and C.~Zhong, ``Statistical {CSI} based design for intelligent
  reflecting surface assisted {MISO} systems,'' \emph{Science China Information
  Sciences}, vol.~63, no.~12, pp. 1--10, 2020.

\bibitem{ZhaoWu21}
M.-M. Zhao, Q.~Wu, M.-J. Zhao, and R.~Zhang, ``Intelligent reflecting surface
  enhanced wireless networks: Two-timescale beamforming optimization,''
  \emph{IEEE Transactions on Wireless Communications}, vol.~20, no.~1, pp.
  2--17, 2021.

\bibitem{GanZhong21}
X.~Gan, C.~Zhong, C.~Huang, and Z.~Zhang, ``{RIS}-assisted multi-user {MISO}
  communications exploiting statistical {CSI},'' \emph{IEEE Transactions on
  Communications}, vol.~69, no.~10, pp. 6781--6792, 2021.

\bibitem{ParkHeath18}
S.~Park and R.~W. Heath, ``Spatial channel covariance estimation for the hybrid
  {MIMO} architecture: A compressive sensing-based approach,'' \emph{IEEE
  Transactions on Wireless Communications}, vol.~17, no.~12, pp. 8047--8062,
  2018.

\bibitem{ParkAli19}
S.~Park, A.~Ali, N.~Gonz{\'a}lez-Prelcic, and R.~W. Heath, ``Spatial channel
  covariance estimation for hybrid architectures based on tensor
  decompositions,'' \emph{IEEE Transactions on Wireless Communications},
  vol.~19, no.~2, pp. 1084--1097, 2019.

\bibitem{WangZhang19}
Y.~Wang, Y.~Zhang, Z.~Tian, G.~Leus, and G.~Zhang, ``Super-resolution channel
  estimation for arbitrary arrays in hybrid millimeter-wave massive {MIMO}
  systems,'' \emph{IEEE Journal of Selected Topics in Signal Processing},
  vol.~13, no.~5, pp. 947--960, 2019.

\bibitem{ElOmar14}
O.~El~Ayach, S.~Rajagopal, S.~Abu-Surra, Z.~Pi, and R.~W. Heath, ``Spatially
  sparse precoding in millimeter wave {MIMO} systems,'' \emph{IEEE Transactions
  on Wireless Communications}, vol.~13, no.~3, pp. 1499--1513, 2014.

\bibitem{VieringHofstetter02}
I.~Viering, H.~Hofstetter, and W.~Utschick, ``Spatial long-term variations in
  urban, rural and indoor environments,'' in \emph{Proceedings of the 5th
  Metting COST}.\hskip 1em plus 0.5em minus 0.4em\relax Lisbon, Pirtugal, 2002.

\bibitem{AlkhateebEl14}
A.~Alkhateeb, O.~El~Ayach, G.~Leus, and R.~W. Heath, ``Channel estimation and
  hybrid precoding for millimeter wave cellular systems,'' \emph{IEEE Journal
  of Selected Topics in Signal Processing}, vol.~8, no.~5, pp. 831--846, 2014.

\bibitem{LiChi15}
Y.~Li and Y.~Chi, ``Off-the-grid line spectrum denoising and estimation with
  multiple measurement vectors,'' \emph{IEEE Transactions on Signal
  Processing}, vol.~64, no.~5, pp. 1257--1269, 2015.

\bibitem{NegahbanRavikumar12}
S.~N. Negahban, P.~Ravikumar, M.~J. Wainwright, and B.~Yu, ``A unified
  framework for high-dimensional analysis of $ {M} $-estimators with
  decomposable regularizers,'' \emph{Statistical Science}, vol.~27, no.~4, pp.
  538--557, 2012.

\bibitem{SoZhang07}
A.~M.-C. So, J.~Zhang, and Y.~Ye, ``On approximating complex quadratic
  optimization problems via semidefinite programming relaxations,''
  \emph{Mathematical Programming}, vol. 110, no.~1, pp. 93--110, 2007.

\bibitem{HaghighatshoarCaire16}
S.~Haghighatshoar and G.~Caire, ``Massive {MIMO} channel subspace estimation
  from low-dimensional projections,'' \emph{IEEE Transactions on Signal
  Processing}, vol.~65, no.~2, pp. 303--318, 2016.

\bibitem{ParkPark17}
S.~Park, J.~Park, A.~Yazdan, and R.~W. Heath, ``Exploiting spatial channel
  covariance for hybrid precoding in massive {MIMO} systems,'' \emph{IEEE
  Transactions on Signal Processing}, vol.~65, no.~14, pp. 3818--3832, 2017.

\bibitem{AnjinappaGurbuz20}
C.~K. Anjinappa, A.~C. G{\"u}rb{\"u}z, Y.~Yap{\i}c{\i}, and
  I.~G{\"u}ven{\c{c}}, ``Off-grid aware channel and covariance estimation in
  mmwave networks,'' \emph{IEEE Transactions on Communications}, vol.~68,
  no.~6, pp. 3908--3921, 2020.

\bibitem{YuXu19}
X.~Yu, D.~Xu, and R.~Schober, ``{MISO} wireless communication systems via
  intelligent reflecting surfaces,'' in \emph{2019 IEEE/CIC International
  Conference on Communications in China (ICCC)}.\hskip 1em plus 0.5em minus
  0.4em\relax Changchun, China, 2019, pp. 735--740.

\bibitem{BuneaXiao15}
F.~Bunea and L.~Xiao, ``On the sample covariance matrix estimator of reduced
  effective rank population matrices, with applications to f{PCA},''
  \emph{Bernoulli}, vol.~21, no.~2, pp. 1200--1230, 2015.

\end{thebibliography}


\end{document}